%% file: main_version.tex
\documentclass[11pt,twoside]{article}

\usepackage{xcolor}

\usepackage{fullpage}

\usepackage{epsf}
\usepackage{fancyhdr}
\usepackage{graphics}
\usepackage{graphicx} 
\usepackage{float} 
\usepackage{subfigure} 
\usepackage{psfrag}
\usepackage{comment}

\usepackage[linesnumbered,ruled]{algorithm2e}
\DontPrintSemicolon	

\usepackage{color}
\usepackage{amsthm}
\usepackage{amsfonts}
\usepackage{amsmath}
\usepackage{bm}
\usepackage{amssymb,bbm}
\usepackage[numbers]{natbib}
\usepackage{algorithmic}
\usepackage[usestackEOL]{stackengine}

\usepackage{url}
\usepackage[colorlinks=True,linkcolor=magenta,citecolor=blue,urlcolor=blue,pagebackref=true,backref=page]
{hyperref}
\renewcommand*{\backref}[1]{\ifx#1\relax \else Page #1 \fi}
\renewcommand*{\backrefalt}[4]{%
    \ifcase #1 \footnotesize{(Not cited.)}%
    \or        \footnotesize{(Cited on page~#2.)}%
    \else      \footnotesize{(Cited on pages~#2.)}%
    \fi}

\usepackage{nicefrac}

\usepackage{chngpage}

\usepackage{tabularx}%

\usepackage{enumitem}
\usepackage{booktabs}
\usepackage{pbox}

\usepackage{caption}

\usepackage{mathtools}

\usepackage{fullpage}
\allowdisplaybreaks
\input{final_macros.tex}

\SetKwInput{KwInput}{Input}                
\SetKwInput{KwOutput}{Output}              

\newcommand{\Var}{\text{Var}}

\newcommand{\hattheta}{\widehat{\theta}}

\theoremstyle{remark}

\begin{document}
\title{Martingale Posteriors from Score Functions}
\author{Fuheng Cui \& Stephen G. Walker \\ \\
Department of Statistics and Data Sciences \\
The University of Texas at Austin\\
email: fuheng.cui@austin.utexas.edu, s.g.walker@math.utexas.edu
}
\date{}
\maketitle
\begin{abstract}
Uncertainty associated with statistical problems arises due to what has not been seen as opposed to what has been seen. 
Using probability to quantify the uncertainty the task is to construct a probability model for what has not been seen conditional on what has been seen.
The traditional Bayesian approach is to use prior distributions for constructing the predictive distributions, though recently a novel approach has used density estimators and the use of martingales to establish convergence of parameter values.
In this paper we reply on martingales constructed using score functions. 
Hence, the method only requires the computing of gradients arising from parametric families of density functions. 
A key point is that we do not rely on Markov Chain Monte Carlo (MCMC) algorithms, and that the method can be implemented in parallel. We present the theoretical properties of the score driven martingale posterior. 
Further, we present illustrations under different models and settings.
\end{abstract}



\section{Introduction}
\label{sec:introduction}

The motivation for this paper is a novel idea about the quantification of uncertainty with an increasing size of population. To make the opening discussions concrete we think of a population mean corresponding to a normal model with a known variance. If the sample size is $n$, and if this is known to coincide with the population size, then the representation of uncertainty is a probability point mass of 1 at the observed sample mean; i.e. $\overline{x}_n$. We now ask how does this point mass change as the true population size grows from $m\geq n$ up to infinity. The case of infinity would represent the population size for the usual Bayesian posterior. For a normal model we are effectively asking for a distribution on $\overline{x}_m$ given $x_{1:n}$ for all $m>n$. Write this is $\Pi_m^{(n)}$. So $\Pi_n^{(n)}$ is always the point mass at $\overline{x}_n$. A posterior distribution for an infinite population will be $\Pi_\infty^{(n)}$, provided the distribution of the limit of the $(\Pi_m^{(n)})$ exists. 

{Let $\theta$ be the parameter of interest with data model $p(x\mid\theta)$, $x \in \mathcal{X}$, $\theta\in\Theta$.} The Bayesian approach would provide $\Pi_m^{(n)}$ from the predictive model which itself comes from the posterior $\pi(\theta\mid x_{1:n})$, derived in the usual way as prior times likelihood. Then,
$$p(x_{n+1:m}\mid x_{1:n})=\int\prod_{i=n+1}^m \mathcal{N}(x_i\mid\theta,\sigma^2)\,\pi(\theta\mid x_{1:n})\,d\theta.$$
This clearly defines $\Pi_m^{(n)}$ since this predictive model provides a distribution for $\overline{x}_m$ given $x_{1:n}$.

It was a result of \cite{Doob_1949} which established that
$\Pi_\infty^{(n)}(\cdot)\equiv \pi(\cdot\mid x_{1:n})$, providing what was originally thought to be a Bayesian consistency theorem. But it is more than this since it indicates an equal importance for the sequence of predictive distributions as with the sequence of posterior distributions.
To make ideas concrete, we will in the Introduction derive the $(\Pi_m^{(n)})$ for the normal model as a first illustration.

Recent research has supported alternative ideas for generating $p(x_{n+1:m}\mid x_{1:n})$ which has simpler constructions to the Bayesian predictive model. See \cite{Fong_2021} and \cite{Holmes}.
In this paper we approach the set up for the predictives to be more in keeping with the goal for the $\Pi_m^{(n)}$ rather than maintaining any particular properties of convenience for the $x_{m+1:n}$, such as exchangeability or conditionally identically distributed; the former leading to a Bayesian set up and the latter to a martingale posterior set up. 

However, due to the desire that the expected value for each $\Pi_m^{(n)}$ be $\overline{x}_n$, we require that the sequence
$(\overline{x}_m)$ be a martingale. This also aids with the convergence properties as $m\to\infty$.
Therefore, we actually define the sequence $(\Pi_m^{(n)})$ via a sequence of random parameter values which in general represent the statistic of interest based on a population of size $m$ with an observed sample of size $n$. The distribution of $\theta_m$ given $x_{1:n}$ is precisely $\Pi_m^{(n)}$.  

It is best here to return to the normal model and write down how the Bayesian version of the sequence works and how our approach is constructed and then comparisons made. 

Let us first simplify the Bayesian model to an objective posterior, so with a flat prior it is that
$\pi(\theta\mid x_{1:n})=\mathcal{N}(\theta\mid \overline{x}_n,\sigma^2/n).$
For the Bayesian predictive model, the predictive sample is given by
$$x_{n+1}=\overline{x}_n+z_1\sigma/\sqrt{n}+z_2\,\sigma,$$
where the first two terms are the sample of a $\theta$ from the posterior and the final term is from the normal model with variance $\sigma^2${, and $z_1$ and $z_2$ are sampled from $\mathcal{N}(0,1)$ independently}. Hence, 
$\overline{x}_{n+1}$ has mean $\bar{x}_n$ and with variance 
$\sigma^2(1+1/n)$. For $\overline{x}_m$, it is easy to see that
the mean is still $\overline{x}_n$, and hence we have a martingale, and the variance is
$$\sigma^2\sum_{i=n+1}^m \frac{1}{i(i-1)}.$$
Note how this increases with $m$ and that in the limit with $m=\infty$ we recover the variance of 
$$\sigma^2\sum_{i=n+1}^\infty\left\{\frac{1}{i-1}-\frac{1}{i}\right\}=\sigma^2/n,$$
which is returning the same variance as with $\pi(\cdot\mid x_{1:n})$, a demonstration of Doob's result.

Our approach, in this special case, would involve taking
$x_{m+1}=\bar{x}_m+z\,\sigma.$
The motivation for this, which extends to the more general models we consider throughout the paper, is that we are dispensing with a posterior component. The difference, as we will show, is negligible and no more than a difference which would occur between two Bayesians using different priors.
We have the $(\overline{x}_m)$ remaining as a martingale, though now the variance of $\overline{x}_m$ is
$$\sigma^2\,\sum_{i=n+1}^m \frac{1}{i^2}.$$
The difference between this variance and that with the Bayesian variance is seen to be of order $1/n^2$. Actually, the difference between posteriors of sample sizes $n$ using different priors will be of order $1/n^2$. This is the order of difference between two Bayesians with different priors, usually comprising a combination of the difference in variances and the difference in means squared.

Our approach to the construction of sequences is based on score functions. To see exactly how, we start by showing a sequence of means $(\theta_m)_{m>n}$ constructed from an observed mean estimator $\theta_n$, based on a sample of size $n$. We therefore introduce what we name the Bayesian score function.
Using the predictive Bayesian model,
$$\theta_{m+1}=\frac{\int \theta\,p(x_{m+1}\mid\theta)\,\pi(\theta\mid x_{1:m})\,d\theta}
{\int p(x_{m+1}\mid\theta)\,\pi(\theta\mid x_{1:m})\,d\theta}$$
which we write as
$$\theta_{m+1}=\theta_m+\frac{\int (\theta-\theta_m)\,p(x_{m+1}\mid\theta)\,\pi(\theta\mid x_{1:m})\,d\theta}
{\int p(x_{m+1}\mid\theta)\,\pi(\theta\mid x_{1:m})\,d\theta}.$$
Hence, we define the Bayesian score function generated by probability density function $\pi$ as
$$s_\pi(x,\theta)=\frac{\int (\gamma-\theta)\,p(x\mid \gamma)\,\pi(\gamma)\,d \gamma}
{\int p(x\mid \gamma)\,\pi(\gamma)\,d \gamma},$$
where $\theta=\int \gamma\pi(\gamma)\,d\gamma$.
So
$\int s_{\pi}(x,\theta)\,p(x\mid\theta)\,\pi(\theta)\,d\theta=0.$
For large $n$ we can recover the Taylor expansion of $p(x\mid \gamma)$ about $\theta$, which, under suitable regularity conditions appearing later, yields
$$p(x\mid \gamma)=p(x\mid \theta)+(\gamma-\theta)p'(x\mid \theta)+\frac{1}{2} (\gamma-\theta)^2\,p''(x\mid \theta)+o(1/n).$$
Hence,
$$s_{\pi_n}(x,\theta)=\sigma^2_n\,\frac{p'(x\mid \theta)}{p(x\mid \theta)}+o(1/n)$$
where $\sigma_n^2=\int(\gamma-\theta)^2\pi_n(\gamma)\,d \gamma$ is the posterior variance. 

An update of the parameters for $m>n$ would then proceed as
$$\theta_{m+1}=\theta_m+\sigma^2_m\,\frac{p'(x_{m+1}\mid \theta_m)}{p(x_{m+1}\mid \theta_m)},$$
where, in order to secure the martingale, $x_{m+1}$ is coming from $p(\cdot\mid \theta_m)$. The sequence of variances are not now available due to the non-existence of a posterior distribution though clearly a sequence of variances can be  determined based on the notion that variances behave as the reciprocal of the sample size; i.e. $1/m$. Therefore, the update is of the form
\begin{equation}\label{sga}
\theta_{m+1}=\theta_m+\sigma^2_m\,s(x_{m+1},\theta_m)
\end{equation}
with $x_{m+1}\sim p(\cdot\mid \theta_m)$. We will explore choices for $(\sigma_m^2)$ later in the paper, though note it is exactly a step size for a gradient ascent/descent algorithm for which there is an abundance of literature.

\begin{center}
\begin{figure}[!htbp]
\begin{center}
\includegraphics[width=\linewidth]{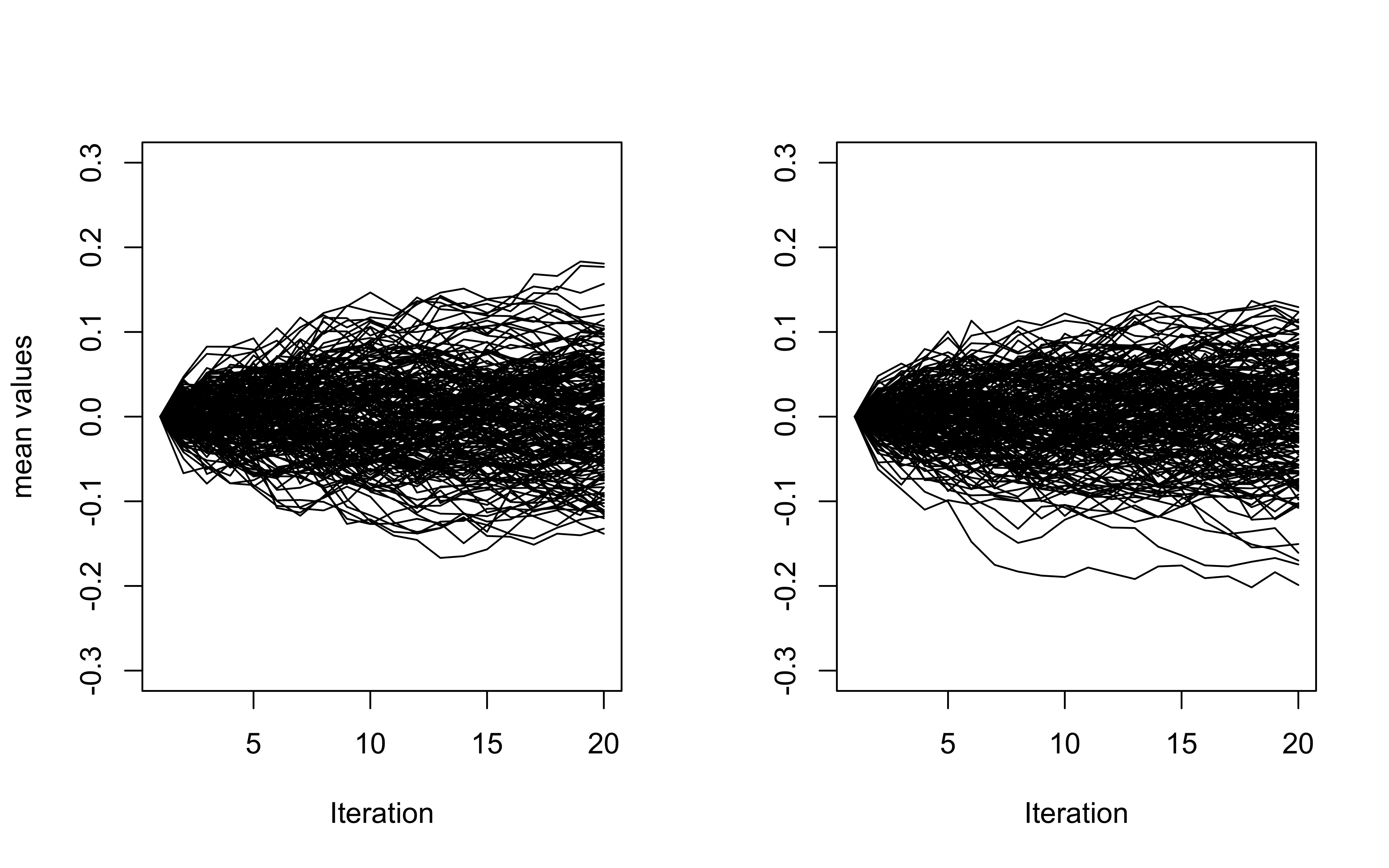}
\caption{Demonstration of Doobs' martingale representation of Bayesian posterior (left panel) and score function martingale (right panel). The iterations are those from $n$ onwards and all trajectories start at 0. }
\label{fig1}
\end{center}
\end{figure}
\end{center}  

To illustrate with a toy example, we consider the normal model with unknown mean $\theta$ and known variance 1. The traditional objective Bayesian posterior distribution is $\mathcal{N}(\cdot\mid\bar{x}_n,1/n)$, and we illustrate with $n=50$ and $\bar{x}_n=0$. 
The left panel demonstrates Doob's result, that the posterior means converge to a $\theta_\infty$ and that these come from the posterior.
There are 200 trajectories.
The right panel is 200 trajectories of the martingale based on the score function with $\sigma^2_m=1/(m+1)$.


{Prior to our development of martingale posteriors based on score functions, extensive research has focused on inference and sampling methods for complex posterior distributions. Notable approaches include variational Bayes (VB) and approximate Bayesian computation (ABC). VB can be viewed as a natural extension of the EM algorithm for posterior approximation; see, for instance, a brief tutorial in~\cite{Tran_2021}. ABC, on the other hand, offers a means of bypassing the evaluation of complicated likelihood functions; see \cite{Sisson_2007}, \cite{Toni_2009} and \cite{Stephan_2017} for example.} The layout of the paper is as follows: In Section~\ref{sec:bayesian_martingale_bootstrap}, we show the methodology of martingale posteriors from score functions and the theoretical guarantees in both one- and multi- dimensional cases. 
We present some illustrations under different settings and different models, including exponential families and autoregressive flows in Section~\ref{sec:illustration}. We also compare our martingale posteriors from the score functions with posteriors from {the standard Bayesian approach and} the variational Bayes (VB) approach on both simulated and real datasets. Section~\ref{sec:summary_and_discussion} provides a summary and discussion.

The contribution of this paper can be stated as follows:
\begin{description}
\item  1. We propose a novel prior-free method to construct posterior distributions. 

\item 2. We rely on the notion of asymptotic exchangeability for missing observations, rather than the current Bayesian criterion of exchangeability for both the observed and missing observations.

\item 3. We employ a novel use for stochastic gradient algorithms, with martingales, with the idea of constructing convergent sequences of parameters where the data has been imputed from the predictive model with the current parameter values. 

\item 4. The limit of the parameter sequences represent a sample from the posterior distribution. 

\item 5. As a consequence, we do not need MCMC algorithms and can run the martingales in parallel.


\end{description}

\section{Parametric martingales}
\label{sec:bayesian_martingale_bootstrap}

In this section we motivate the score function gradient ascent algorithm for updating the parameters.
As a first illustration, we show that the score function update is a first order approximation to a Bayesian update. To see this and to elaborate on the Bayesian score function introduced in the Section~\ref{sec:introduction},
$$p(x_{m+1}\mid \theta)=p(x_{m+1}\mid\theta_m)+(\theta_m-\theta)\,p'(x_{m+1}\mid\theta_m)+\frac{1}{2} (\theta_m-\theta)^2\,p''(x_{m+1}\mid\widetilde\theta_m),$$
for some $\widetilde\theta_m$ lying between $\theta$ and $\theta_m$, where $\theta_m$ is the posterior mean based on a sample of size $m$, and $'$ denotes differentiation with respect to $\theta$.
The Bayesian update of the posterior mean is
$$\theta_{m+1}=\frac{\int \theta\,p(x_{m+1}\mid\theta)\,\pi(\theta\mid x_{1:m})\,d\theta}{p_{m}(x_{m+1})}$$
where $p_m(x)=\int p(x\mid\theta)\,\pi(\theta\mid x_{1:m})\,d\theta$.
Hence, some regular Taylor series expansions imply that
$$\theta_{m+1}=
\theta_m+\mbox{Var}(\theta_m\mid x_{1:m})\,\frac{p'(x_{m+1}\mid\theta_m)}{p(x_{m+1}\mid\theta_m)}+o(1/m),$$
if the model is suitably regular and the posterior variance is as usual of order $1/m$. Of course, here
$p'(x\mid \theta)/p(x\mid\theta)=s(x,\theta)$, the score function.
Therefore, the $\epsilon_m$ in (\ref{sga}) could refer to the Bayesian posterior variance with a sample of size $m$. We will look at this idea in more detail later in the paper.

The score function is defined as the gradient of the log-likelihood function with respect to the parameter:
$
s(x,\theta) = \nabla_{\theta} \log p(x\mid\theta).
$
An important property of the score function is that its expectation is equal to 0 at the value of the parameter, under some regularity conditions; i.e.
$
E_{x\sim p(\cdot\mid\theta)}\,s(x,\theta) = 0.
$
The suitable regularity conditions have been studied by a number of various authors. In this paper, we use and modify the conditions given in Chapter 4 of~\cite{Serfling_2009}.

The basic algorithm for constructing the Bayesian parametric 
martingale posteriors $(\Pi_m^{(n)})$ for $m>n$, is given by
\begin{align}
    x_{m+1}  &\sim p(\cdot\mid  \widehat\theta_{m})\label{equ:algo_resample},\\
     \widehat\theta_{m+1} = \hattheta_{m} &+ \epsilon_{m}s(x_{m+1},\widehat\theta_{m})\label{equ:algo_update},
\end{align}
{where $\widehat{\theta}_m$ is the estimate of $\theta$ with the sample of size $m$, i.e. $\widehat{\theta}_m \sim \Pi_m^{(n)}$}. For the step size, we rely on the use of the Fisher information matrix, when available and computable. When not, for say larger models, we can find a suitable step size using ideas for setting such from the stochastic gradient algorithms literature; see \cite{Schaul_2013}, \cite{Nguyen_2019}, \cite{Granziol_2022}, \cite{Wang_2023} and \cite{Shi_2023}, for example.
First we will look at the one dimensional case followed by the multivariate case.


\subsection{The one dimensional case}
\label{subsec:one_dim_cases}
Here we assume $\theta \in \Theta \subset \mathbb{R}$, and the collected data set of size $n$ is $x_{1:n}$. In order to establish convergence of the martingale (\ref{sga}) for $m\geq n$, we 
consider the sequence of variances, and find conditions under which it is bounded. From the properties of conditional variances, we have
\begin{align}
    \Var(\hattheta_{m+1}) &= \Var (\hattheta_m)+E\{\epsilon_m^2 \Var(s(x_{m+1},\hattheta_m))\}\label{equ:var_tower_rule_1}\\
    &=\sum_{i=n}^m E\{\epsilon_i^2 \,\Var(s(x_{i+1},\hattheta_i))\}\label{equ:var_tower_rule_2}.
\end{align}
Here the $E\{\mbox{Var}\}$ term has the variance acting on the $x_{m+1}$ and the expectation on $\widehat\theta_i$. It is well known that $\mbox{Var}(s(x,\theta))=I(\theta)$, the Fisher information. Hence,
$$\mbox{Var}(\widehat\theta_m)=\sum_{i=n}^{m-1}E\{\epsilon_i^2\,I(\widehat\theta_i)\}.$$
We can set a suitable variance by taking 
$\epsilon_i=\tau_i/\sqrt{I(\widehat\theta_i)}$ so $$\Var(\hattheta_{\infty})=\sum_{i=n}^{\infty}\tau_i^2.$$ 
To obtain a standard asymptotic posterior variance of 
$(n\,I(\widehat\theta_n))^{-1}$, see \citep{Vaart_1998}, it would be appropriate to take
$$\tau_i=\frac{1}{(1+i)\sqrt{I(\widehat\theta_n})}.$$
Hence, in this case $\mbox{Var}(\widehat\theta_\infty)<\infty$
and hence $\widehat\theta_\infty<\infty$ a.s. More generally, here we establish regularity conditions under which $\hattheta_\infty$ exists; i.e. $\Pi_{\infty}^{(n)}$ exists.  

\vspace{0.5 em}
\noindent
The following is following the conditions to be found in \cite{Serfling_2009}: Let $\Theta$ to be an open interval in $\mathbb{R}$. We assume:

\begin{enumerate}[label=(RO.1) ,align= left]
\item\label{regularity:existance} 
For each $\theta \in \Theta$, the derivative $\partial \log p(x \mid \theta)/\partial \theta$ exists for all $x$.
\end{enumerate}

\begin{enumerate}[label=(RO.2) ,align= left]
\item\label{regularity:differentiable} 
There exists a function \(g(x)\), such that $\left|\partial p(x \mid \theta)/\partial \theta\right| \leq g(x)$
holds for all $x$ and \(\theta \in \Theta\), and $\int g(x) d x<\infty$.
\end{enumerate}

\begin{enumerate}[label=(RO.3.a) ,align= left]
\item\label{regularity:finite_uniform_bound}
There exists a fixed number $M>0$, such that
\begin{align*}
    0<E_{\theta}\left\{\left(\frac{\partial \log p(X \mid \theta)}{\partial \theta}\right)^{2}\right\}\leq M
\end{align*}
for all $\theta \in \Theta$. 
\end{enumerate}

\begin{enumerate}[label=(RO.3.b) ,align= left]
\item\label{regularity:finite_step_sizes}
For each $\theta \in \Theta$, $E_{\theta}\left\{\left({\partial \log p(X \mid \theta)}/{\partial \theta}\right)^{2}\right\}$ exists. Set $\epsilon_i=\tau_i/\sqrt{I(\widehat\theta_{i})}$ where the series $\{\tau_i\}_{i=n}^\infty$ satisfy $\sum_{i=n}^{\infty}\tau_i^2<\infty$.
\end{enumerate}

\vspace{0.5 em}
\noindent
We can show the convergence of $\{\hattheta_m\}$ as stated in the following theorem.
\begin{theorem}
\label{theorem:convergence}
Consider the Bayesian martingale posterior scheme with score functions ~\ref{equ:algo_resample} and ~\ref{equ:algo_update}. 
\begin{enumerate}
    \item \label{theorem:convergence_finite}
    If the regularity conditions~\ref{regularity:existance} and~\ref{regularity:differentiable} are satisfied, and either~\ref{regularity:finite_uniform_bound} or \ref{regularity:finite_step_sizes} is satisfied, then the $\{\hattheta_i\}$ sequence will converge to a finite random variable $\hattheta_{\infty}$ almost surely, with $\Var(\hattheta_{\infty})<\infty$.
    
    \item \label{theorem:convergence_bounded}
    If the regularity conditions~\ref{regularity:existance},~\ref{regularity:differentiable} and~\ref{regularity:finite_uniform_bound} are satisfied, and $\epsilon_i=1/i$, then $\Var(\hattheta_{\infty})\leq M/(n-1)$.
    
    \item \label{theorem:convergence_accurate}
    If the regularity conditions~\ref{regularity:existance},~\ref{regularity:differentiable} and~\ref{regularity:finite_step_sizes} are satisfied, and $\tau_i=L/i$, where $L$ is a constant, then
    \begin{align*}
        L/n \leq \Var(\hattheta_{\infty})\leq L/(n-1).
    \end{align*}
\end{enumerate}
\end{theorem}

\begin{proof}
    From the Lebesgue Dominant Convergence Theorem, and Condition \ref{regularity:existance} and \ref{regularity:differentiable}, we have
    \begin{align}
        E\left(\frac{\partial \log p(X \mid \theta)}{\partial \theta}\right) &= \int \frac{\partial \log p(x\mid  \theta)}{\partial \theta} p(x\mid\theta)\mathrm{d}x\\
        &=\int \frac{1}{p(x\mid\theta)}\frac{\partial  p(x \mid \theta)}{\partial \theta}p(x\mid\theta)\mathrm{d}x\\
        &=\frac{\partial}{\partial \theta}\int p(x\mid\theta)\mathrm{d}x=0.
    \end{align}
    Therefore, $E(\widehat\theta_{m+1}\mid\hattheta_{m})=\widehat\theta_m$ for all $m\geq n$.
    
    Moreover, $\Var\left(\partial \log p(X \mid \theta)/\partial \theta\right)=E_{\theta}\left\{\left(\partial \log p(X \mid \theta)/\partial \theta\right)^{2}\right\}$ if they exist. Therefore, by Condition~\ref{regularity:finite_uniform_bound} or \ref{regularity:finite_step_sizes} and (\ref{equ:var_tower_rule_1}) and (\ref{equ:var_tower_rule_2}), we have
    $$\Var(\hattheta_{m+1})=\sum_{i=n}^m \epsilon_i^2 \Var\left(s(X_{i+1},\hattheta_i)\right)<\infty$$
    for every $m\geq n$. Therefore, $\{\hattheta_m\}_{m=n}^{\infty}$ is a martingale, and by Doob's martingale convergence theorem, see \cite{Doob_1949}, we have that $\hattheta_m$ will converge to a finite random variable $\hattheta_{\infty}$ almost surely with $\Var(\hattheta_{\infty})<\infty$.

    If Condition~\ref{regularity:finite_uniform_bound} is satisfied and we choose $\epsilon_i=1/i,$ we will have
    $$
    \Var(\hattheta_{m+1})\leq\sum_{i=n}^{m}M/i^2
    $$
    for any $m\geq n$.
    By the Monotone Convergence Theorem, we have
    $$
    \Var(\hattheta_{\infty})\leq\sum_{i= n}^{\infty}M/i^2\leq M/(n-1).
    $$
Similarly, if Condition~\ref{regularity:finite_step_sizes} is satisfied and $\tau_i={L}/{i},$ for any $m\geq n$, we have
    $
    \Var(\hattheta_{m+1})=\sum_{i=n}^{m}L/i^2,
    $
    and by the Monotone Convergence Theorem,
    $
    \Var(\hattheta_{\infty})=\sum_{i=n}^{\infty}L/i^2.
    $
    Therefore, we have  $L/n \leq \Var(\hattheta_{\infty}) \leq L/(n-1).$
\end{proof}

\subsection{Multiple dimensional case}
\label{subsec:multi_dim_cases_vector_step}
In this section, we assume $\theta \in \Theta \subset \mathbb{R}^d$, $\{\epsilon_i\}\in \mathbb{R}^d$. Starting with an estimator of the parameter of interest $\hattheta_n = \hattheta(X_1,\ldots,X_n)$, and with the density $p(x\mid\theta)$, the score function $s(x,\theta)$, a vector of dimension $d$, and the step-size $\epsilon_m$, for $m \geq n$, also a vector of size $d$, the updates arise as 
\begin{align}
     x_{m+1} &\sim p(x\mid \hattheta_{m}) \label{equ:algo_resample_vector},\\
     \hattheta_{m+1} &= \hattheta_{m} + \epsilon_{m}\otimes s(x_{m+1},\hattheta_{m}),\label{equ:algo_update_vector}
\end{align}
where $\otimes$ is the Kronecker product.

We can derive that if each component of $\theta$ satisfies the regularity conditions ~\ref{regularity:existance}, ~\ref{regularity:differentiable} and ~\ref{regularity:finite_uniform_bound} or ~\ref{regularity:finite_step_sizes}, the whole vector will also satisfies the part~\ref{theorem:convergence_finite} in Theorem~\ref{theorem:convergence}.

\vspace{0.5 em}
\noindent
Consider $\Theta$ to be an open interval in $\mathbb{R}^d$, $\epsilon_m\in\mathbb{R}^d$ for each $m\geq n$ and the update rule~(\ref{equ:algo_resample_vector}) and~(\ref{equ:algo_update_vector}). We assume:

\begin{enumerate}[label=(RM.1) ,align= left]
\item\label{regularity:multi_existance_vector} 
For each $\theta \in \Theta$, the gradient $\nabla_{\theta} \log p(x \mid \theta)$ exists, for all $x$.
\end{enumerate}

\begin{enumerate}[label=(RM.2) ,align= left]
\item\label{regularity:multi_differentiable_vector} 
There exists a function \(g(x)\), such that for each $k=1,\ldots,d$, $\left|\partial p(x \mid \theta)/\partial \theta_k\right| \leq g(x)$
holds for all $x$ and \(\theta \in \Theta\), and $\int g(x) d x<\infty$.
\end{enumerate}

\begin{enumerate}[label=(RM.3.a) ,align= left]
\item\label{regularity:multi_finite_uniform_bound_vector}
There exists a fixed number $R>0$, such that the Fisher Information
\begin{align*}
    0<\trace\left( E_{\theta}\left\{\left(\frac{\partial \log p(X \mid \theta)}{\partial \theta}\right)\left(\frac{\partial \log p(X \mid \theta)}{\partial \theta}\right)^\intercal\right\}\right)\leq R
\end{align*}
for all $\theta \in \Theta$. 
\end{enumerate}

\begin{enumerate}[label=(RM.3.b) ,align= left]
\item\label{regularity:multi_finite_step_sizes_vector}
For each $\theta \in \Theta$, $E_{\theta}\left\{\left(\frac{\partial \log p(X \mid \theta)}{\partial \theta}\right)\left(\frac{\partial \log p(X \mid \theta)}{\partial \theta}\right)^\intercal\right\}$ exists. Set $$\mathcal{D}(\theta)=\diag\left(E_{\theta}\left\{\left(\frac{\partial \log p(X \mid \theta)}{\partial \theta}\right)\left(\frac{\partial \log p(X \mid \theta)}{\partial \theta}\right)^\intercal\right\}\right)$$ and $\epsilon_{m,k}=\tau_{m,k}/\sqrt{\mathcal{D}(\hattheta_{m})_k}$, $k=1,\ldots,d$. The series $\{\tau_{m,k}\}_{m=n}^\infty$ satisfies $\sum_{m=n}^{\infty}\tau_{m,k}^2<\infty$ for all $k=1,\ldots,d$.
\end{enumerate}

\vspace{0.5 em}
\noindent
The convergence of $\{\hattheta_m\}_{m\geq n}$ is stated in the following theorem.
\begin{theorem}
\label{theorem:multi_convergence_vector}
Consider the Bayesian martingale {posterior given by} (\ref{equ:algo_resample_vector}) and (\ref{equ:algo_update_vector}):
\begin{enumerate}
    \item \label{theorem:convergence_finite_vector}
    If the regularity conditions~\ref{regularity:multi_existance_vector} and~\ref{regularity:multi_differentiable_vector} are satisfied, and either~\ref{regularity:multi_finite_uniform_bound_vector} or~\ref{regularity:multi_finite_step_sizes_vector} is satisfied, then the $\{\hattheta_m\}$ sequence will converge to a finite random vector $\hattheta_{\infty}$ almost surely, with all elements of the variance matrix $\Var(\hattheta_{\infty})$ finite.
    
    \item \label{theorem:convergence_bounded_vector}
    If the regularity conditions ~\ref{regularity:multi_existance_vector}, ~\ref{regularity:multi_differentiable_vector} and ~\ref{regularity:multi_finite_uniform_bound_vector} are satisfied, and $\epsilon_{m,k}={1}/{m}$, then $\Var(\hattheta_{\infty})_{k,k}\leq {R}/({n-1})$ for all $k=1,\ldots,d$.
    
    \item \label{theorem:convergence_accurate_vector}
    If the regularity conditions ~\ref{regularity:multi_existance_vector}, ~\ref{regularity:multi_differentiable_vector} and ~\ref{regularity:multi_finite_step_sizes_vector} are satisfied, and $\tau_{m,k}=L_k/m$ where $L_k$ are constants only depending on $k$, then 
    \begin{align*}
        L_k/n \leq \Var(\widehat\theta_{\infty})_{k,k} \leq L_k/(n-1)
    \end{align*}
    for all $k=1,\ldots,d$.
\end{enumerate}
\end{theorem}

\begin{proof}
    The proof that $\{\hattheta_{m,j}\}_{m=n}^{\infty}$ is a martingale for every $j=1,\ldots,d$ is the same as that for Theorem~\ref{theorem:convergence}, where $\hattheta_{m,j}$ means the $j$-th element of $\hattheta_{m}$. For every $j$, $\hattheta_{m,j}$ will converge to a random variable $\hattheta_{\infty,j}$ as $m\to\infty$ with a finite expectation and variance. So every $\hattheta_{m,j}\hattheta_{m,k}$ also converges to a random variable with a finite expectation for every $j$ and $k$. Since $\{\mathrm{Cov}(X,Y)\}^2\leq\Var(X)\Var(Y)$ holds for any random variables $X$ and $Y$, the variance-covariance matrix $\Var(\hattheta_{\infty})$ is finite. Therefore, the process of the random vector is a martingale and $\{\hattheta_m\}$ is convergent almost surely by Doob's martingale convergence theorem. 
    
    The remainder of the proof 
    is identical as that for Theorem~\ref{theorem:convergence}, when considering the vector process element-wise. 
\end{proof}

By choosing  suitable step sizes, we can get a posterior with mean as the original estimator and the variance to be $(n I(\widehat\theta_n))^{-1}$, matching the Bernstein-von Mises theorem, which was first established in~\cite{Doob_1949}. The Bernstein-von Mises theorem is an asymptotic results requiring a large sample size $n$ for the approximate normality of the posterior. Our method obtains the same mean and variance for all $n$  without considering normality. See~\cite{Vaart_1998} for more details about the Bernstein-von Mises theorem.

\subsection{Asymptotic exchangeability}
\label{subsec:Asymptotic Exchangeability}

Having established that $\widehat\theta_m\to\widehat\theta_\infty$ a.s., we can define $\mathcal{P}_m$ to be the probability measure for $\widehat\theta_m$ and $\mathcal{P}_\infty$ the probability measure for $\widehat\theta_\infty$ and $\mathcal{P}_m\to\mathcal{P}_{\infty}$ according to weak convergence of probability measures. See \cite{Berti_2006} and \cite{Lijoi_2007}.  Then by Lemma 8.2(b) in~\cite{Aldous_1985}, $(X_{n+1:\infty})$ where $X_{m+1}\sim p(\cdot\mid\widehat\theta_m)$ is an asymptotically exchangeable sequence, in the sense that there exists an exchangeable sequence $(Z_{1:\infty})$ for which
$$(X_{m+1},X_{m+2},\ldots)\to_d (Z_1,Z_2,\ldots),$$
as $m\to\infty.$ Indeed, the $(Z_i)$ are i.i.d. from $p(\cdot\mid \theta)$ conditional upon $\theta$ with $\theta\sim\mathcal{P}_\infty$. For the mixture model, we may have a sequence of conditionally identically distributed (c.i.d.) predictive samples, which infers the asymptotic exchangeability; see~\cite{Fortini_2020} and~\cite{Cui_2023}.


\section{Illustrations}
\label{sec:illustration}

Here we consider a number of examples, starting with some simple cases and then moving on to multidimensional models 
including a time series model. The illustrations will use the Algorithm~\ref{algo:bayesian_martingale_bootstrap_with_scores} as the means by which to obtain the posterior samples.

\begin{algorithm}[ht!]
\DontPrintSemicolon
  
  \KwInput{Data collected $\{X_1, \ldots, X_n\}$, where $n$ is the sample size; the model is $p(x\mid\theta)$; the initial estimator for the parameter is $\widehat\theta_n$, the score function is $s(x,\theta)$, and the step-size $\epsilon_m=1/(m+1)$ for $m=n,\ldots,T$, where $T$ is length run of each martingale. $D$ will be the number of martingales generated, and hence will also be the number of samples from the posterior.}
  Set $\hattheta_{ln}=\widehat{\theta}_n$ for $l = 1,\ldots,D$\\
  \For{$l= 1,2,\ldots,D,$ \textbf{in parallel}}{
  \For{$m=n,n+1,\ldots,T$}{
  $X_{l,m+1} \sim p(x\mid \widehat\theta_{l,m})$\\
  $\widehat\theta_{l,m+1}=\widehat\theta_{l,m}+\epsilon_{m} s(X_{l,m+1}, \widehat\theta_{l,m})$
  }
  Set $\widehat\theta_l=\widehat\theta_{l,T}$
  }
  \KwOutput{$\{\widehat\theta_l\}_{l=1}^D$}
\caption{Parametric martingale with score function}
\label{algo:bayesian_martingale_bootstrap_with_scores}
\end{algorithm}

\subsection{Exponential family}
\label{subsec:exponential_families}

Consider the family
\begin{align}
    \label{equ:exponential family}
    p(x\mid\theta)=h(x) \exp\left\{\sum_{j=1}^d \theta_j\,\phi_j(x)-\beta(\theta)\right\}
\end{align}
where the $(\phi_j)$ are known functions. 
The score functions are given by
$$s_j(x,\theta)=(\partial/\partial\theta_j)\log p(x\mid \theta)=\phi_j(x)-\partial \beta(\theta)/\partial\theta_j.$$
We then update $\theta_j$ given observation $X_{m+1}=x_{m+1}$ from $p(\cdot\mid\theta_m)$ via
\begin{align}
\label{equ:update_analytic}
    \theta_{m+1,j}=\theta_{m,j}+\epsilon_m\,\left\{\phi_j(x_{m+1})-\left.\frac{\partial \beta(\theta)}{\partial\theta_j}\right|_{\theta=\theta_{m}}\right\},
\end{align}
for some step size $\epsilon_m$. The $\beta_j'(\theta)=\partial \beta(\theta)/\partial\theta_j$ may be difficult to obtain, but it can be arbitrarily well approximated using
$$\widehat{\beta}'_j(\theta)=N^{-1}\sum_{l=1}^N \phi_j(X^*_l),$$
for some large $N$,
where the $(X^*_l = x^*_l)$ are sampled i.i.d. from $p(\cdot\mid\theta_m)$. Hence, we can update, sequentially, over $m\geq n$, and iteratively, over $j=1,\ldots,d$, 
\begin{align}
    \label{equ:update_approx}
    \theta_{m+1,j}=\theta_{m,j}+\epsilon_m\,\left\{\phi_{j}(x_{m+1})-N^{-1}\sum_{l=1}^N \phi_j(x^*_l)\right\}.
\end{align}
Without the covariance matrix update, this would be analogous to a mean field variational style posterior distribution.

Notice that $\phi(x)=(\phi_1(x),\ldots,\phi_d(x))$ in (\ref{equ:exponential family}) is the sufficient statistic for $\theta=(\theta_1,\ldots,\theta_d)$. Therefore, it makes sense that the distribution of the next updated $\theta_m$ only depends on the distribution of the sufficient statistic $\phi(x)$.


\subsubsection{Normal distribution with known variance}
\label{subsubsec:Normal Distribution with the Variance Known}

First we present a simple toy example with the normal distribution with the variance known. 
The update rule from (\ref{equ:update_analytic}) becomes
\begin{align}
    \label{equ:update_normal_mean_analytical}
    \mu_{m+1}=\mu_m+\epsilon_m(X_{m+1}-\mu_m)
\end{align}
with $X_{m+1}\sim \mathcal{N}( \mu_m,\sigma^2)$ for $m>n$.


For the illustration, we collect $n=100$ independent samples from $\mathcal{N}(2,1)$. We use the maximum likelihood estimate $\widehat{\mu}_{100}=\mu_{100}=2.016$ as the starting point for each martingale sequence. The martingale updates for $\mu$ are as in (\ref{equ:update_analytic}) {for which we update the scores by the exact analytical form of $\partial \beta(\theta)/\partial \theta_j$, and (\ref{equ:update_approx}) where we use the Monte Carlo method to approximate $\partial \beta(\theta)/\partial \theta_j$}. The number of posterior samples we take is $D=1000$, and each martingale is run for $T=2000$ iterations. 
See Figure~\ref{pl:Gaussian_mean_varknown} and Table~\ref{tab:Gaussian_mean_varknown} for the results. Both methods have posterior means approximately $\mu_{100}$, as they should, and standard deviation $1/\sqrt{n}\approx0.1$. 

\begin{figure}[ht!]
\centering
\subfigure[]{
\label{pl:Gaussian_mean_analytical}
\includegraphics[width=0.45\textwidth]{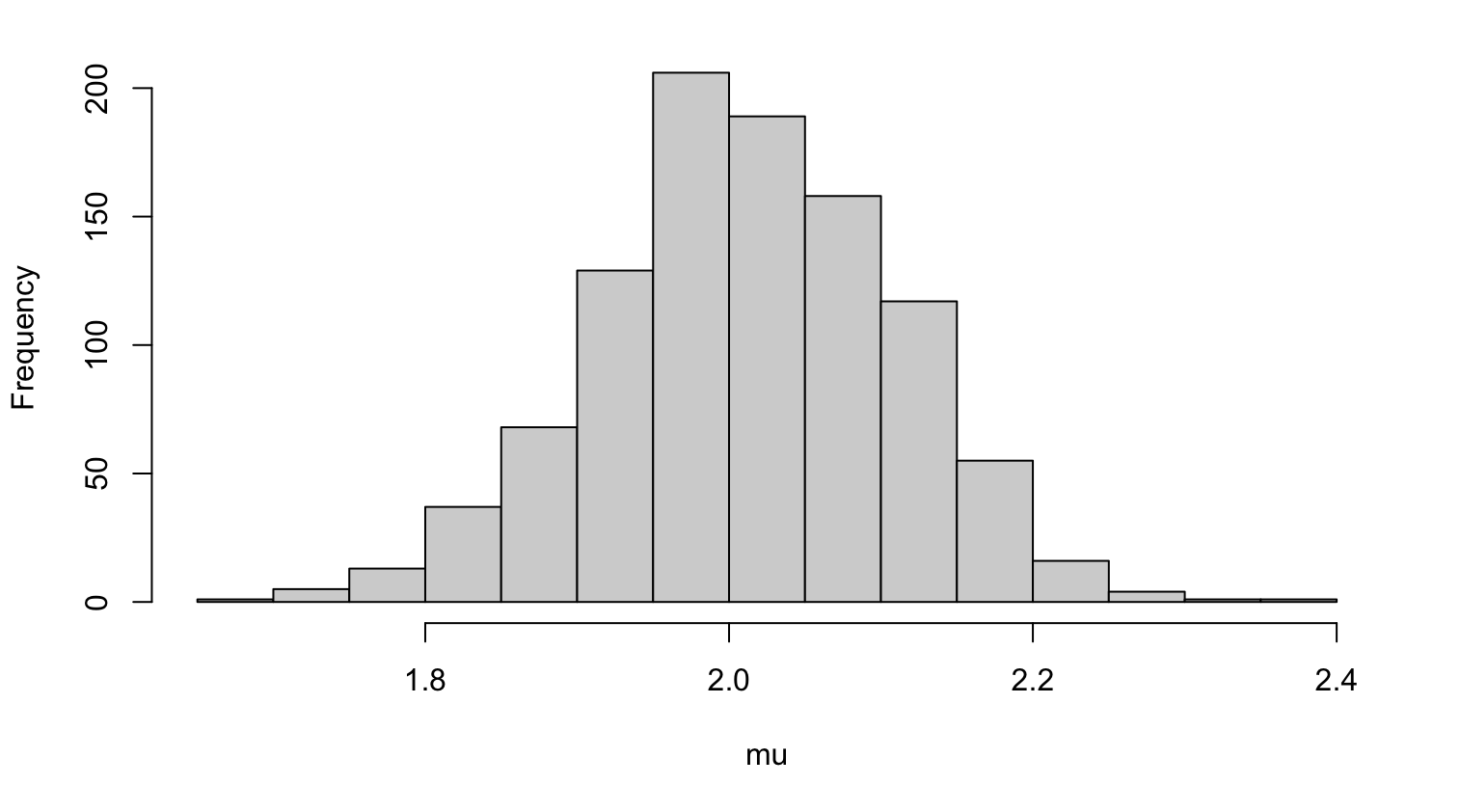}}
\subfigure[]{
\label{pl:Gaussian_mean_approximated}
\includegraphics[width=0.45\textwidth]{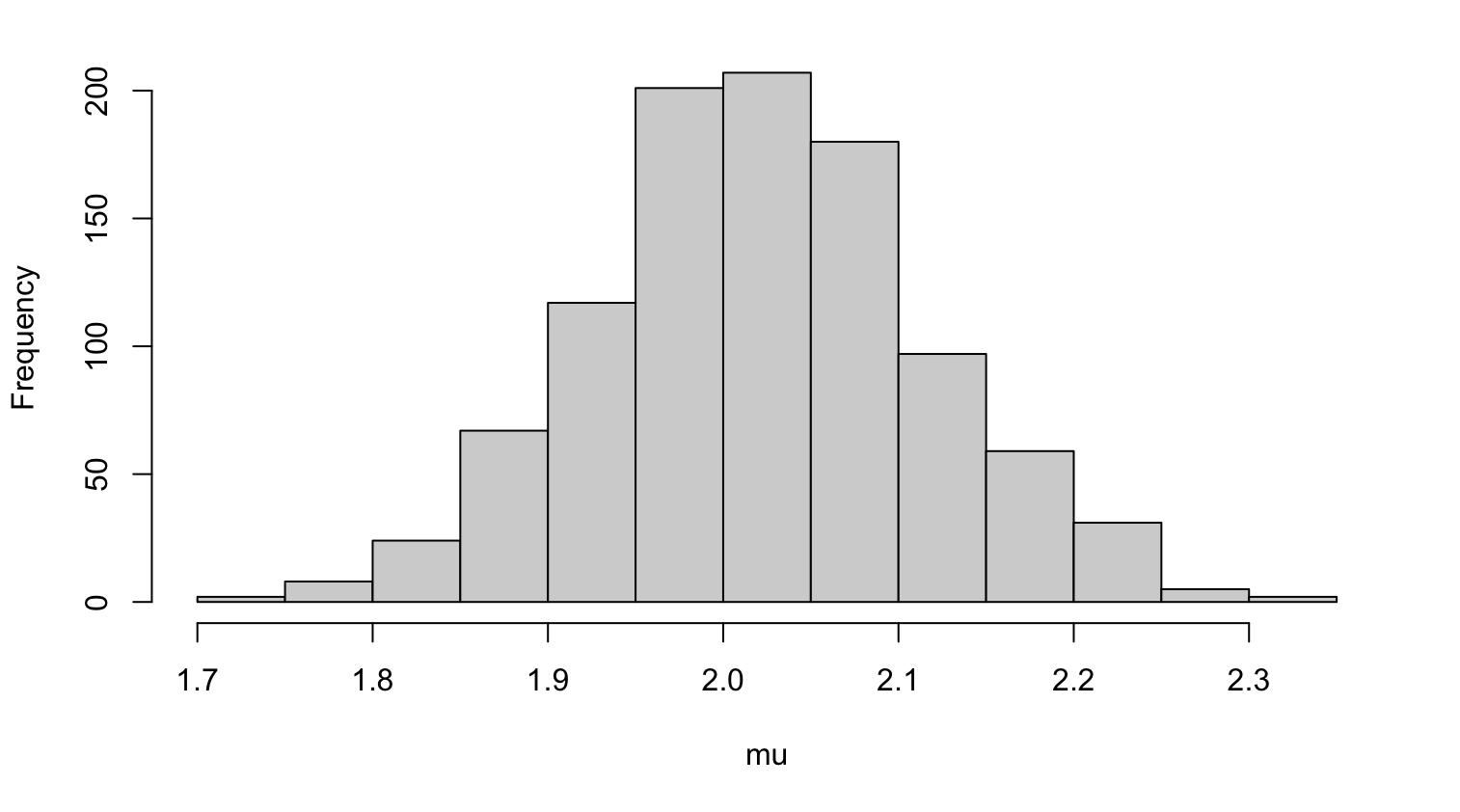}}
\caption{Histograms of posterior samples of $\mu$: (a) the method with exact score update; (b) the method with Monte Carlo approximate score update.}
\label{pl:Gaussian_mean_varknown}
\end{figure}

\begin{table}[ht!]
    \centering
    {\renewcommand{\arraystretch}{2}
    \begin{tabular}{ccc}
    \bottomrule \hline
    Score-update & Mean & Standard Deviation \\
    \Centerstack{Exact } & \Centerstack{2.017} & \Centerstack{0.098} \\
    \Centerstack{Approximate }   & \Centerstack{2.019} & \Centerstack{0.097} \\
    [2ex] \toprule
    \end{tabular}
    }
    \caption{Means and standard deviations for martingale posterior samples of $\widehat{\mu}_{2100}$. {The first row shows the results from the algorithm updated with the exact score functions. The second row presents the results from the algorithm updated by the Monte Carlo approximation to the score functions.}}
    \label{tab:Gaussian_mean_varknown}
\end{table}

\subsubsection{Normal distribution with the mean and variance unknown}
\label{subsubsec:Normal Distribution with the Mean and Variance Unknown}
We now consider the normal distribution with both the mean $\mu$ and variance $\sigma^2$ unknown. Considering the natural parameters $\theta_1$ and $\theta_2$, we have $\theta_1={\mu}/{\sigma^2}$, $\theta_2=-{1}/({2\sigma^2})$, $\phi_1=x$, $\phi_2=x^2$ and $\beta=-{\theta_{1}^{2}}/({4 \theta_{2}})-\log \left(-2 \theta_{2}\right)/2$. With the analytical form of the partial derivatives of $\beta$, the update formula in (\ref{equ:update_analytic}) is
\begin{align}
    \label{equ:update_normal_mean_var_analytical}
    \theta_{m+1,1}&=\theta_{m,1}+\epsilon_m\left\{X_{m+1}-\left(-\frac{\theta_{m,1}}{2\theta_{m,2}}\right)\right\},\\
    \theta_{m+1,2}&=\theta_{m,2}+\epsilon_m\left\{X_{m+1}-\left(\frac{\theta_{m,1}^2}{4\theta_{m,2}^2}-\frac{1}{2\theta_{m,2}}\right)\right\},
\end{align}
with $X_{m+1}\sim \mathcal{N}(\mu_m,\sigma_m^2)$ and $m\geq n$.

Or if using the approximated version with Monte Carlo methods of the partial derivatives of $b$, the update rule (\ref{equ:update_approx}) becomes
\begin{align}
    \label{equ:update_normal_mean_var_approx}
    \theta_{m+1,1}=\theta_{m,1}+\epsilon_m(X_{m+1}-\overline{Z}_m),\\
    \theta_{m+1,2}=\theta_{m,2}+\epsilon_m(X_{m+1}^2-\overline{Z^2}_m),
\end{align}
where $X_{m+1}, Z_1, \ldots Z_N \stackrel{iid}{\sim} \mathcal{N}(\mu_m,\sigma_m^2)$, $\overline{Z}_m={\sum_{i=1}^N Z_i}/{N}$, $\overline{Z^2}_m={\sum_{i=1}^N Z_i^2}/{N}$ and $m\geq n$.

In our experiments, we collect $n=100$ data from $X_{i}\stackrel{iid}{\sim}\mathcal{N}(2,0.5)$. We then use $\widehat{\mu}={\sum_{i=1}^n X_i}/{n}=2.008$ and $$\widehat{\sigma}=\sqrt{{\sum_{i=1}^n (X_i-\widehat{\mu})^2}/({n-1})}=0.503$$ as the starting point for the martingales. We take the number of martingales to be $D=1000$, and each run for  $T=1000$ iterations, which means we sample $1000$ number of $\theta_{1100}$. See Figure~\ref{pl:Gaussian_mv} and Table~\ref{tab:Gaussian_mv} for the results. {As above, we again show the two methods with the scores updated using different forms, in which one uses the exact score function while the other uses a Monte Carlo approximation to the score function.}

\begin{figure}[ht!]
\centering
\subfigure[]{
\label{pl:Gaussian_mv_mu_analytical}
\includegraphics[width=0.45\textwidth]{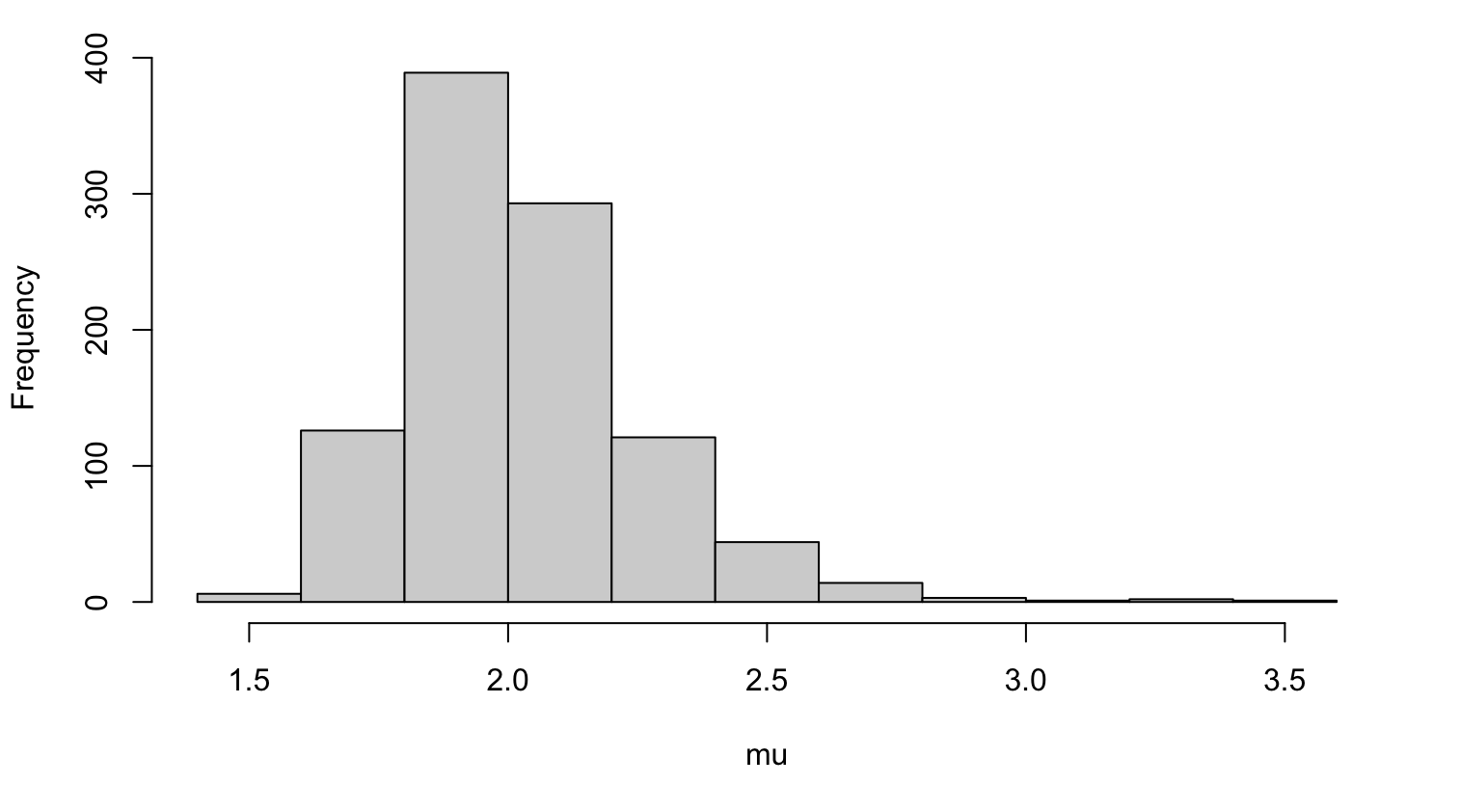}}
\subfigure[]{
\label{pl:Gaussian_mv_sigma_analytical}
\includegraphics[width=0.45\textwidth]{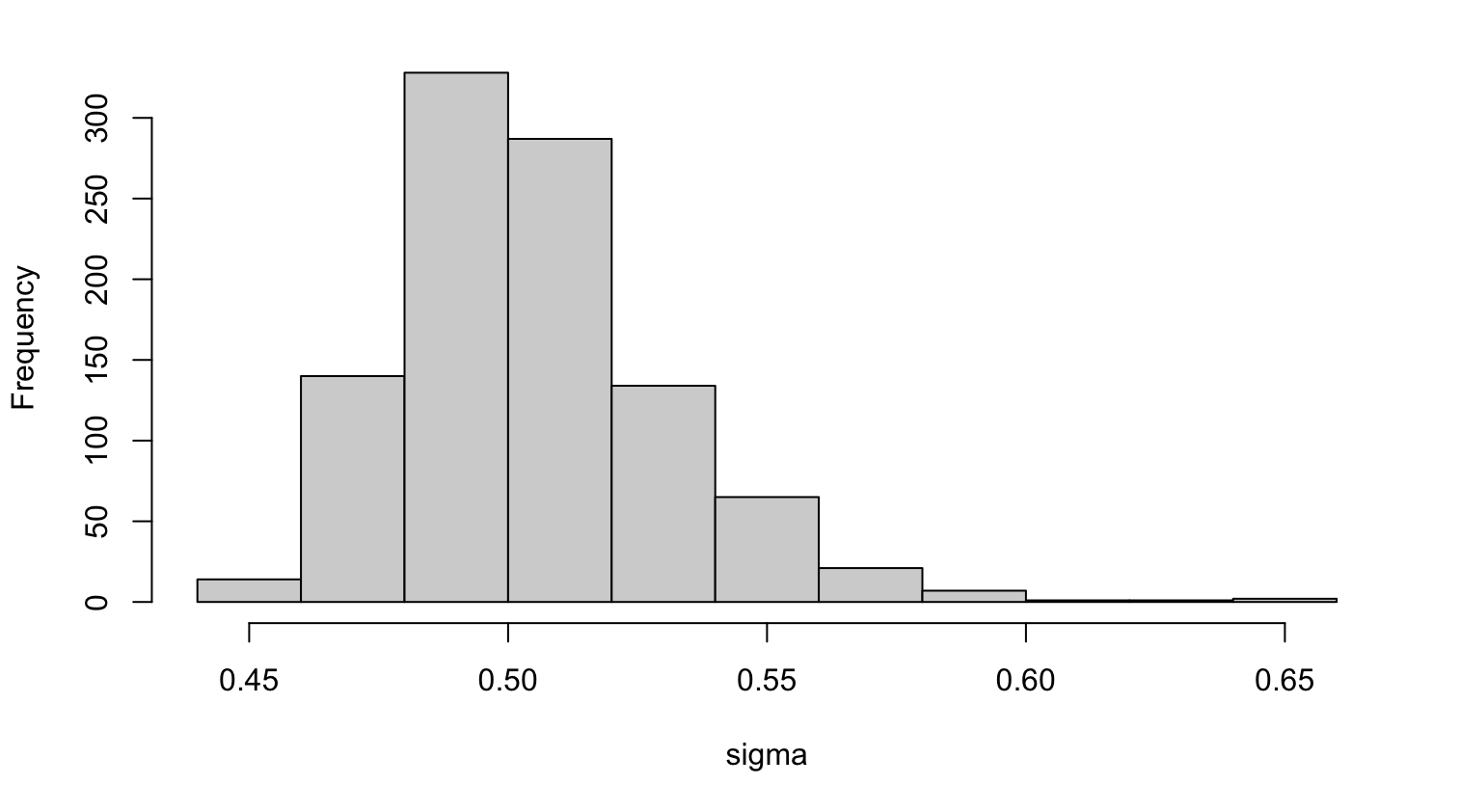}}
\subfigure[]{
\label{pl:Gaussian_mv_mu_approximated}
\includegraphics[width=0.45\textwidth]{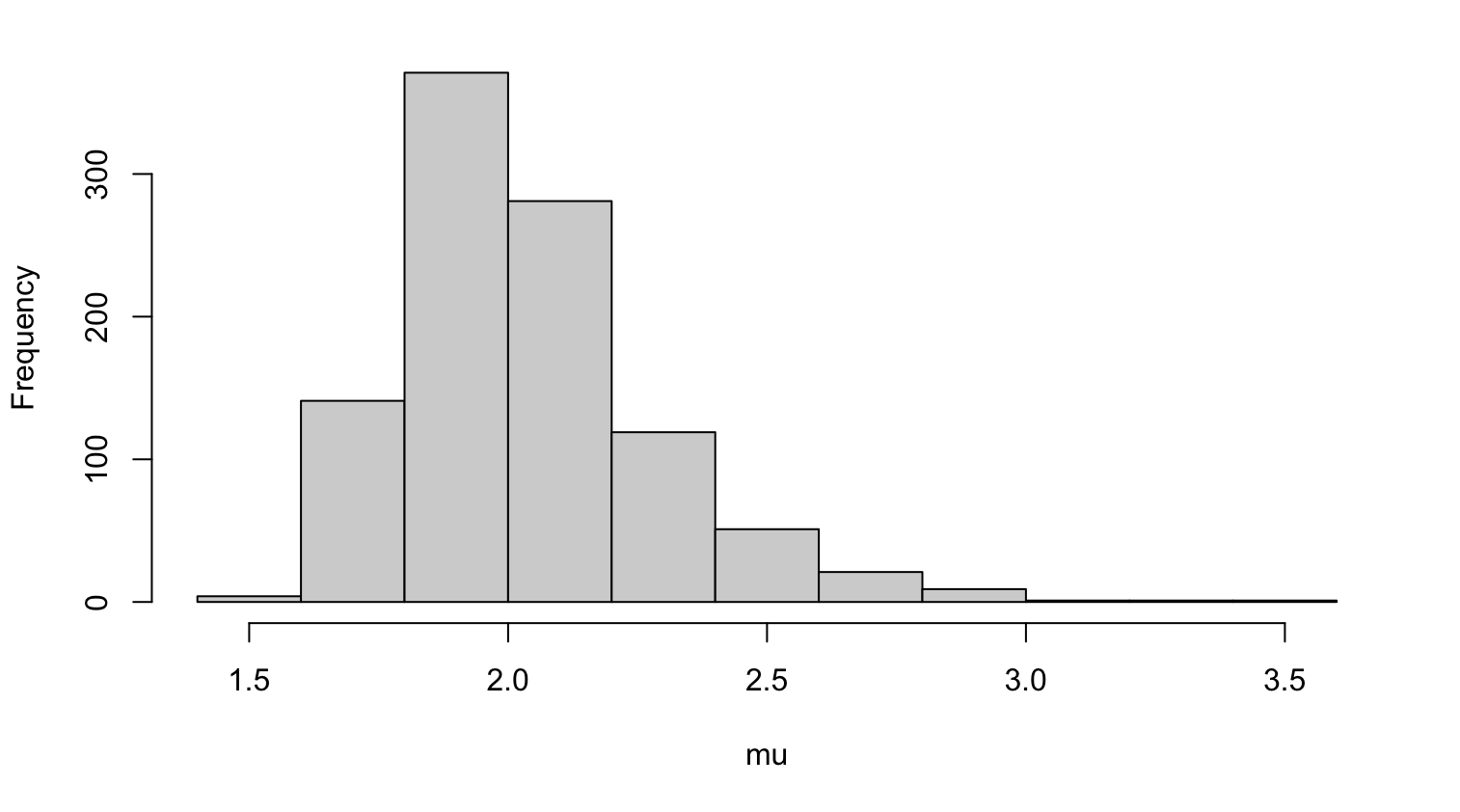}}
\subfigure[]{
\label{pl:Gaussian_mv_sigma_approximated}
\includegraphics[width=0.45\textwidth]{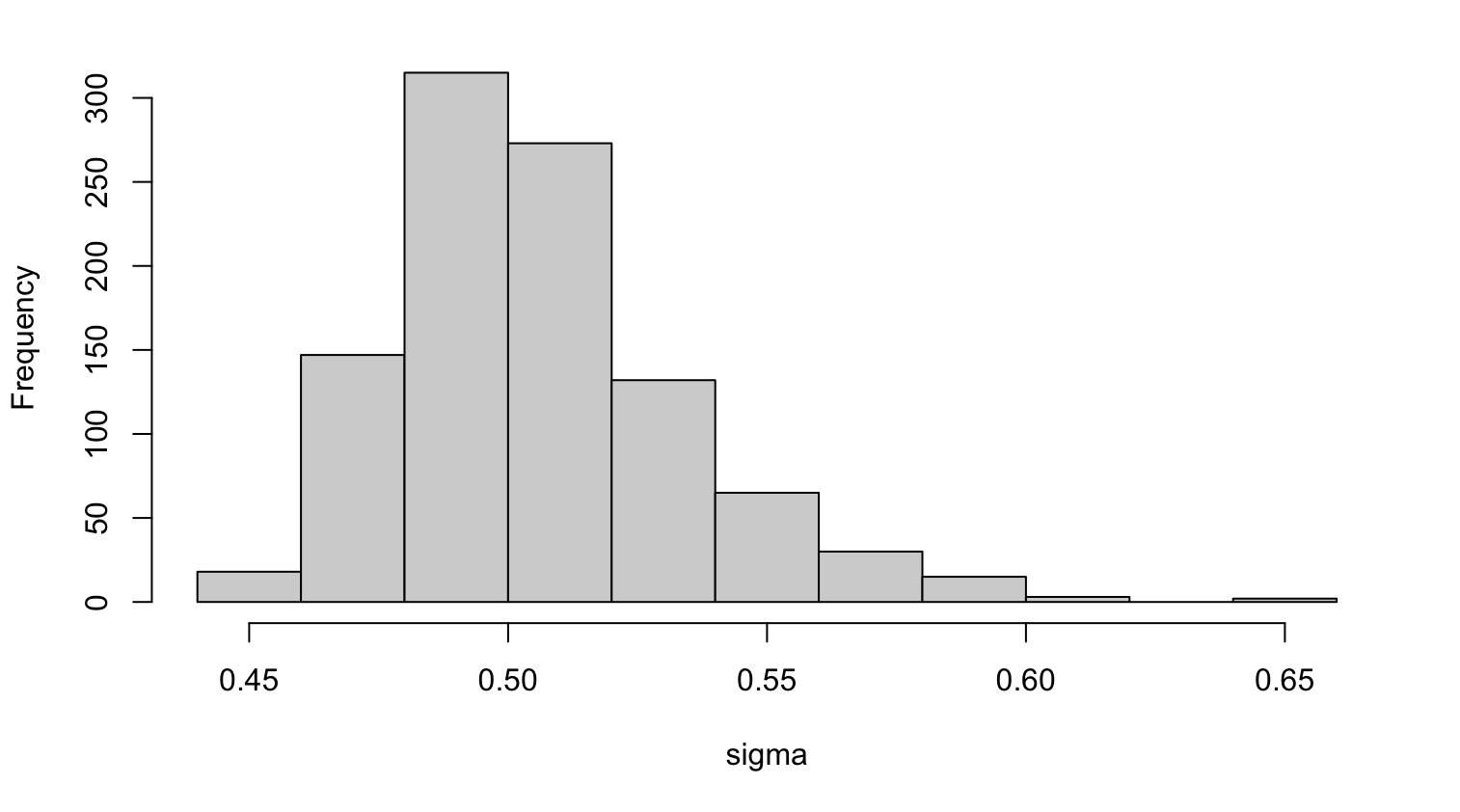}}
\caption{Histograms of martingale posterior samples of $\theta_{1100}$: (a)  mean with analytical ${\partial \beta}/{\partial \theta_j}$; (b) standard deviation with analytical ${\partial \beta}/{\partial \theta_j}$; (c) mean with approximated ${\partial \beta}/{\partial \theta_j}$; (d) standard deviation with approximated ${\partial \beta}/{\partial \theta_j}$.}
\label{pl:Gaussian_mv}
\end{figure}

\begin{table}[ht!]
    \centering
    {\renewcommand{\arraystretch}{2}
    \begin{tabular}{cccc}
    \bottomrule \hline
    Score-update & Parameter &Mean & Standard Deviation \\
    \Centerstack{Exact} &$\mu$ & \Centerstack{2.024} & \Centerstack{0.230} \\
    \Centerstack{Exact} &$\sigma$  & \Centerstack{0.504} & \Centerstack{0.026} \\
    \Centerstack{Approximate} &$\mu$   & \Centerstack{2.033} & \Centerstack{0.250} \\
    \Centerstack{Approximate} &$\sigma$   & \Centerstack{0.505} & \Centerstack{0.029} \\
    [2ex] \toprule
    \end{tabular}
    }
    \caption{Means and standard deviations for martingale posterior samples of $\widehat{\mu}_{1100}$ and $\widehat{\sigma}_{1100}$.}
    \label{tab:Gaussian_mv}
\end{table}

\subsubsection{Beta distribution}
\label{subsubsec:Beta Distribution}
We also consider the Beta distributions, whose partial derivatives of $\beta$ need to estimate the integrals. We use (\ref{equ:update_approx}) to obtain martingale posterior samples. Suppose the distribution is $\mathrm{Beta}(a,b)$ with $a$ and $b$ both unknown, and 
$$
p(x\mid a,b)=\frac{1}{\mathrm{B}(a, b)} x^{a-1}(1-x)^{b-1}
$$
where $\mathrm{B}(a, b)$ is the Beta function and $a>0, b>0$.

We still consider the natural parameters $\theta$: $\theta_1=a, \theta_2=b, \phi_1(x)=\log x, \phi_2(x)=\log(1-x),$ and $\beta(\theta)=\log \Gamma(a)+\log \Gamma(b)-\log \Gamma(a+b)$ where $\Gamma$ is the Gamma function. Therefore, using (\ref{equ:update_approx}) {where we need to use Monte Carlo to approximate the score function}, the update rule is
\begin{align}
    \label{equ:update_beta}
    a_{m+1}&=a_m+\epsilon_m\left(\log X_{m+1}-\overline{\log Z}_m\right),\\
    b_{m+1}&=b_{m}+\epsilon_m\left(\log (1-X_{m+1})-\overline{\log (1-Z)}_m\right),
\end{align}
with $X_{m+1}, Z_1, \ldots, Z_N \stackrel{iid}{\sim} \mathrm{Beta}(a_m,b_m)$ and
\begin{align*}
    \overline{\log Z}_m&=\frac{\sum_{i=1}^N \log Z_i}{N},\\
    \overline{\log (1-Z)}_m&=\frac{\sum_{i=1}^N \log (1-Z_i)}{N}
\end{align*}
where $m\geq n$.

In our experiments, we collect $n=500$ data from $X_{i}\stackrel{iid}{\sim}\mathrm{Beta}(2,1)$. And then we use the methods of moments (MM) to estimate $a$ and $b$, and we get $\widehat{a}=2.106, \widehat{b}=1.033$ as the starting point for each martingale. We {sample from the martingale posterior} on the parameters with $T=1500$ and $D=1000$, which means we sample $1000$ amount of $\theta_{2000}$. See Figure~\ref{pl:Beta} and Table~\ref{tab:Beta} for the results.

\begin{figure}[ht!]
\centering
\subfigure[]{
\label{pl:Beta_a}
\includegraphics[width=0.45\textwidth]{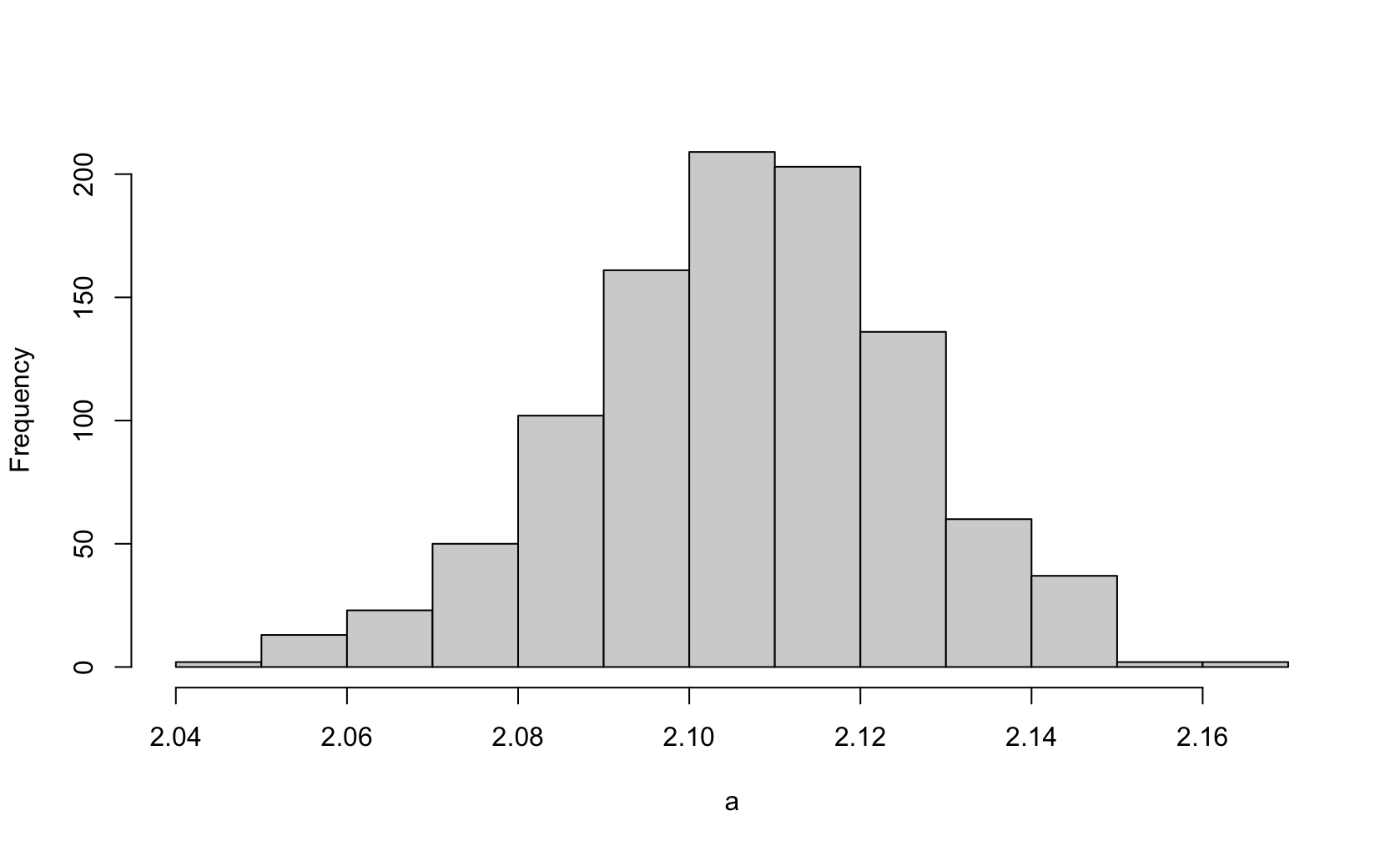}}
\subfigure[]{
\label{pl:Beta_b}
\includegraphics[width=0.45\textwidth]{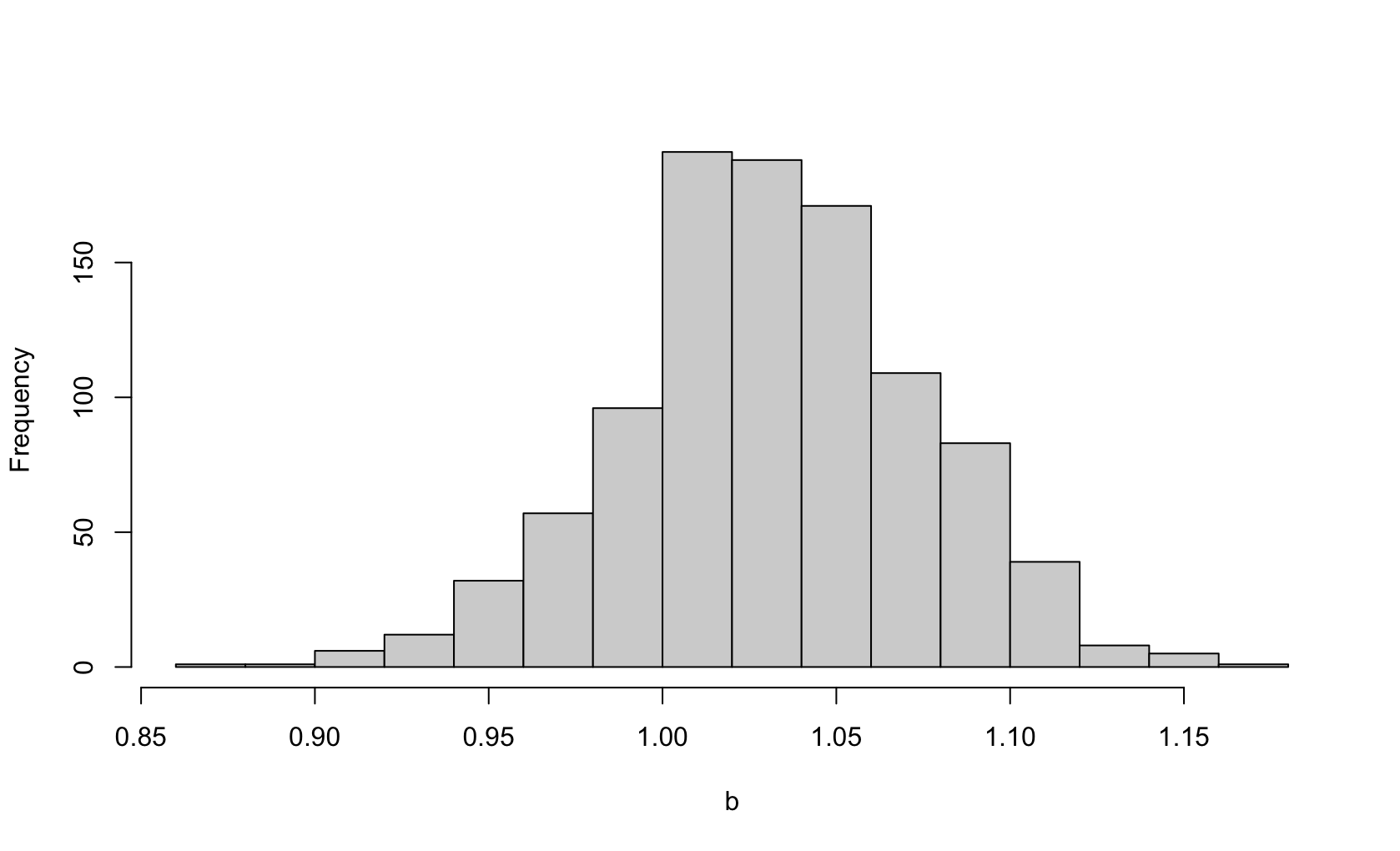}}
\caption{Histograms of martingale posterior samples of $\theta_{2000}$: (a) $a$; (b)  $b$.}
\label{pl:Beta}
\end{figure}

\begin{table}[ht!]
    \centering
    {\renewcommand{\arraystretch}{2}
    \begin{tabular}{ccc}
    \bottomrule \hline
    Parameter & Mean & Standard Deviation \\
    \Centerstack{a} & \Centerstack{2.106} & \Centerstack{0.019} \\
    \Centerstack{b}   & \Centerstack{1.032} & \Centerstack{0.043} \\
    [2ex] \toprule
    \end{tabular}
    }
    \caption{Means and standard deviations for posterior samples of $\theta_{2000}$.}
    \label{tab:Beta}
\end{table}

{The results above demonstrate that our method performs effectively with exponential family distributions. Although deriving the exact analytical form of the score function can sometimes be challenging, Monte Carlo methods can approximate the scores, giving very good results. Our approach successfully constructs the score function and samples efficiently from the corresponding martingale posterior in parallel. In general, performing Bayesian inference on exponential families can be difficult due to the complexity of sampling from the posterior, often needing non-trivial MCMC techniques such as the Metropolis–Hastings algorithm.}

\subsection{Further illustrations}


\subsubsection{Gaussian distribution with mean unknown}
We sample $X_1, \ldots, X_n$ i.i.d. from the distribution $ \mathcal{N}(\theta,1)$. We compare the following two methods: 1) Martingale posterior with update of $\hattheta_n$ using the posterior mean{, which is the MLE in this case and equivalent to the classical Bayesian scheme by Doob's result}; 2) Score function approach choosing the step size $\epsilon_m={1}/({m+1})$ and the score function $s$. Here, we set the true value of parameter $\theta^*=1$, $n=100$, and the {martingale posterior} sample $D=1000$. For each martingale sample, the iteration number is $T=1000$, that is we sample $1000$ of $\hattheta_{1100}$.

\begin{figure}[ht!]
\centering
\subfigure[]{
\label{pl:Gaussian_mean_Doob}
\includegraphics[width=0.45\textwidth]{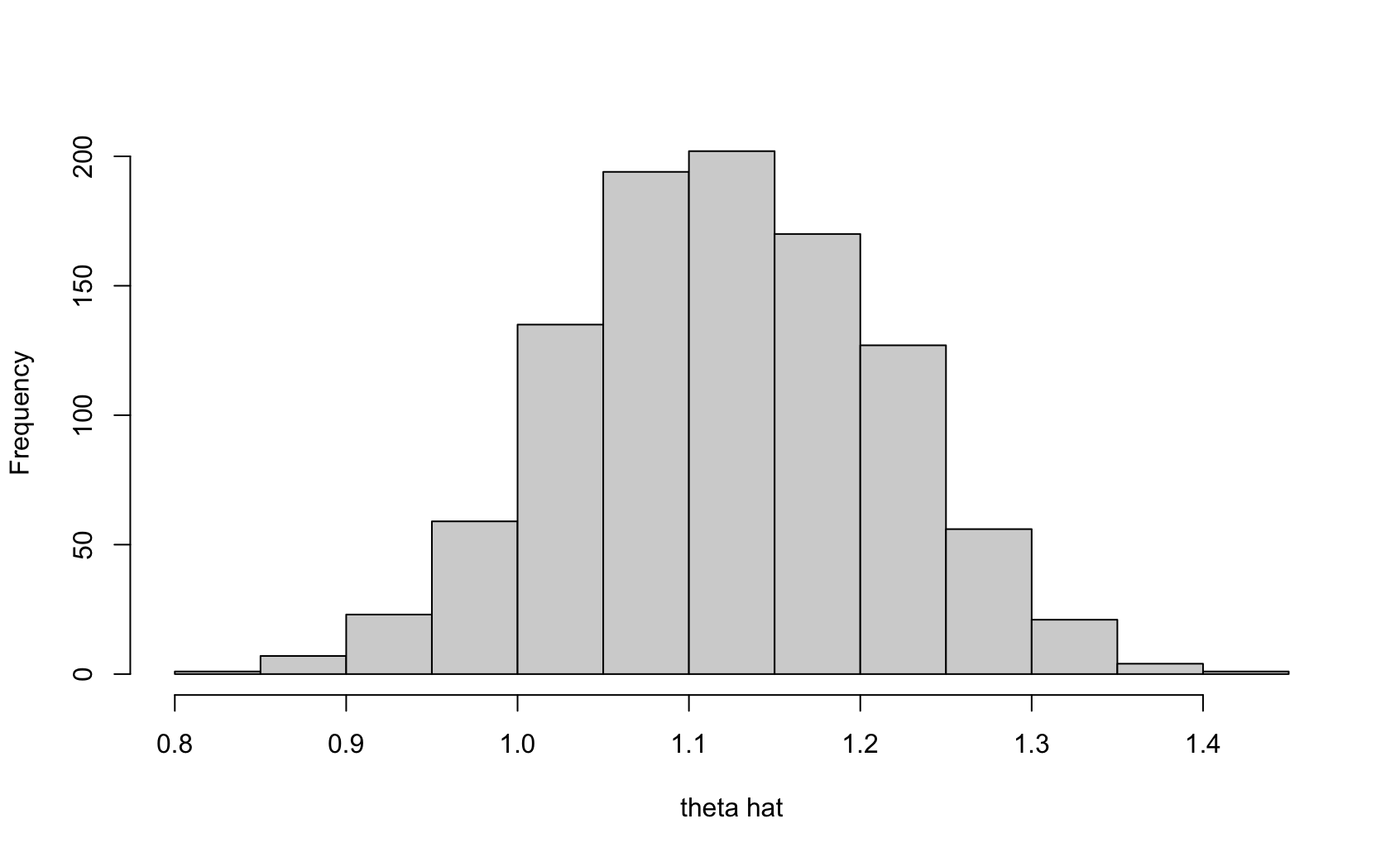}}
\subfigure[]{
\label{pl:Gaussian_mean_our}
\includegraphics[width=0.45\textwidth]{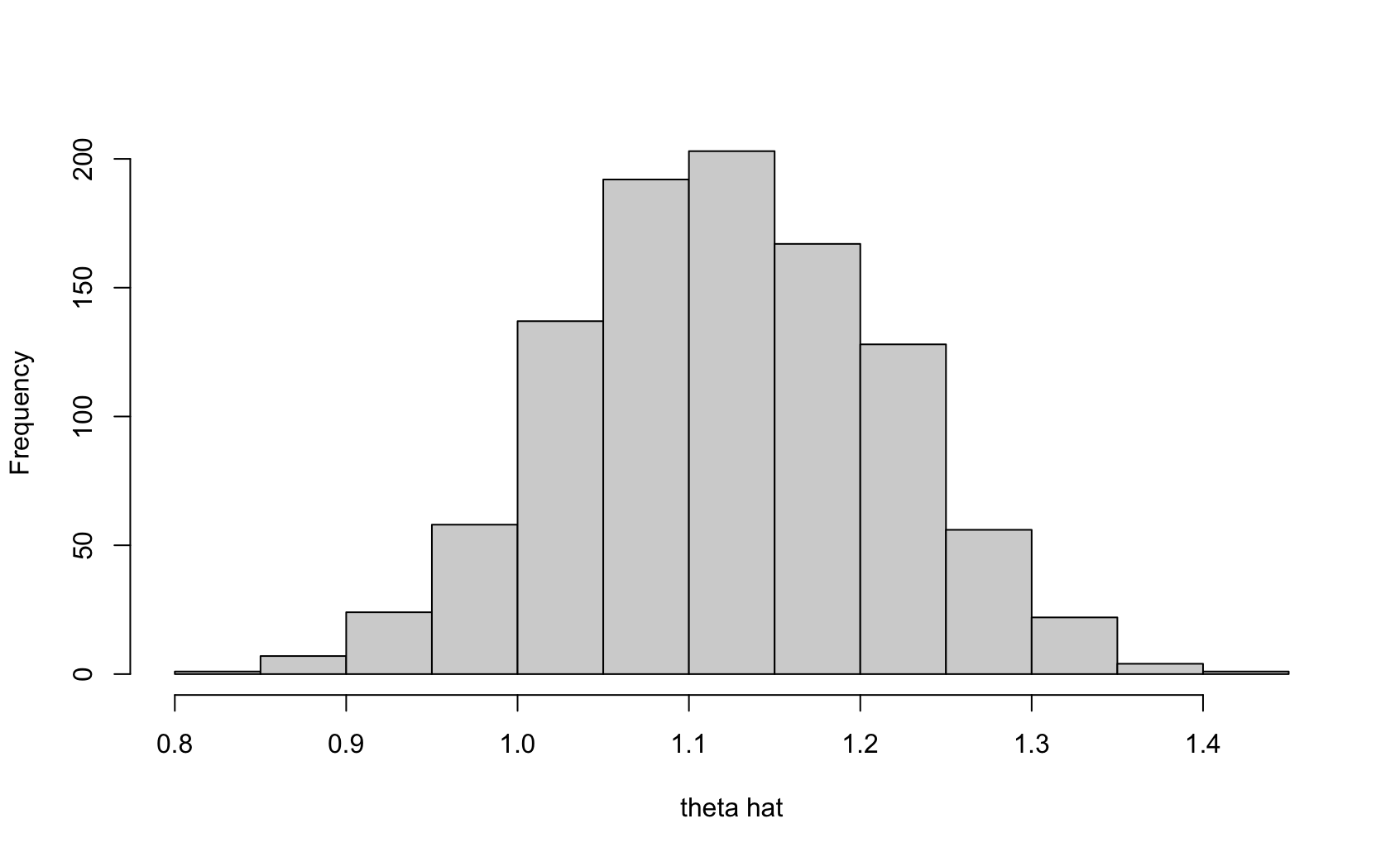}}
\caption{Histograms of posterior samples of $\hattheta_{1100}$: (a) Classical Bayesian method; (b) Score function method.}
\label{pl:Gaussian_mean}
\end{figure}

\begin{table}[ht!]
    \centering
    {\renewcommand{\arraystretch}{2}
    \begin{tabular}{ccc}
    \bottomrule \hline
    Method & Mean & Standard Deviation \\
    \Centerstack{Bayesian} & \Centerstack{1.122} & \Centerstack{0.092} \\
    \Centerstack{Score function}   & \Centerstack{1.122} & \Centerstack{0.093} \\
    [2ex] \toprule
    \end{tabular}
    }
    \caption{Means and deviations for {the martingale posterior} samples of $\hattheta_{1100}$.}
    \label{tab:Gaussian_mean}
\end{table}

\vspace{0.5 em}
\noindent
The relevant histogram is shown in Figure \ref{pl:Gaussian_mean} and the numerical results are given in Table \ref{tab:Gaussian_mean}. We can see the results are similar.

This is the special case when the MLE of $\theta$ is just the sample mean, so with the martingale posterior method we can store the current $\theta_m$ and newly sampled $X_{m+1}$ to update $\theta_{m+1}$ as 
$$\widehat{\theta}_{m+1}=\frac{m\widehat{\theta}_m+X_{m+1}}{m+1}$$
for $m\geq N$.
However, generally $\widehat{\theta}_{m+1}$ will need all collected and resampled data to calculate with the martingale approach, for example, when $\widehat{\theta}_i$ is not sufficient. But with our method using martingales, we only need to store $\theta_m$ and $X_{m+1}$ to get $\theta_{m+1}$ always, which will save storage and time. 

\subsubsection{Gaussian distribution with unknown covariance}
We sample $X_1, \ldots, X_n$ from the distribution $[X|\Sigma] \sim \mathcal{N}(0,\Sigma)$. We compare the following two methods: {1) Classical Bayesian method with predictive sampling scheme; 2) Martingale posterior method with score function approach}, choosing $\epsilon_m={1}/({m+1})$ and the score function $s$. Here, we set $n=10000$ and the number of martingales to be $D=1000$. For each martingale, the iteration number is $T=1000$, that is we sample $D=1000$ of $\hattheta_{11000}$. For the IW prior, we use IW$(3I_5,7)$. Now we check the posterior distribution for (2,2)-element of the precision matrix (true value 4/3), which is the inverse of the covariance, due to the convenience for constructing martingales. In the experiment, the true $\theta^*$ is
$$
\theta^*=
{\begin{bmatrix}
2 & 1 & 0 & 0 & 0\\
1 & 2 & 1 & 0 & 0\\
0 & 1 & 2 & 1 & 0\\
0 & 0 & 1 & 2 & 1\\
0 & 0 & 0 & 1 & 2\\
\end{bmatrix}}^{-1}.
$$

\begin{figure}[ht!]
\centering
\subfigure[]{
\label{pl:Gaussian_covaraince_Doob_22}
\includegraphics[width=0.45\textwidth]{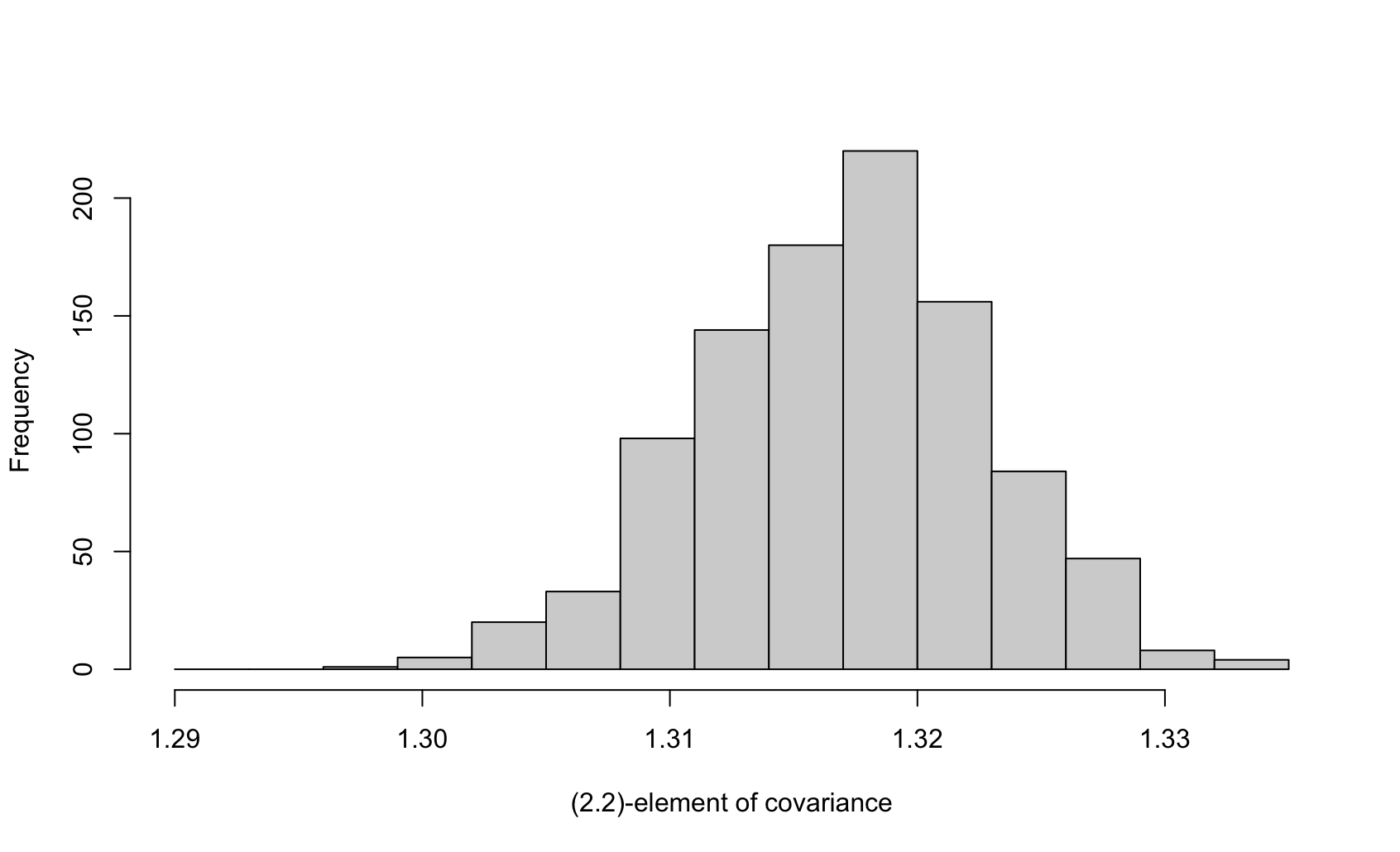}}
\subfigure[]{
\label{pl:Gaussian_covaraince_Doob_IW_22}
\includegraphics[width=0.45\textwidth]{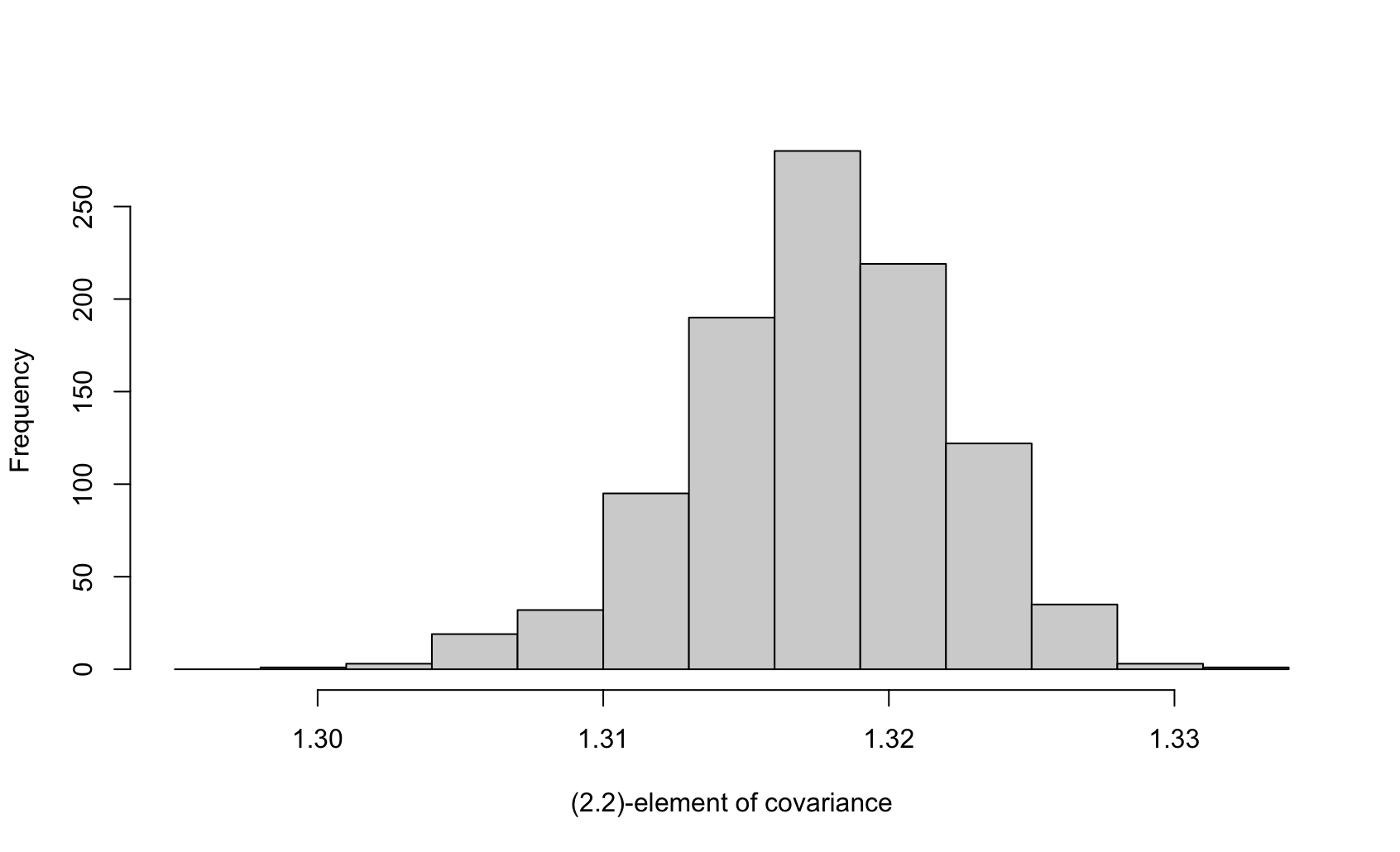}}
\caption{Histograms of martingale posterior samples of $\hattheta_{11000}$: (a) Classical Bayesian method; (b) Score method.}
\label{pl:Gaussian_covariance}
\end{figure}

\begin{table}[ht!]
    \centering
    {\renewcommand{\arraystretch}{2}
    \begin{tabular}{ccc}
    \bottomrule \hline
    Method & Mean & Standard Deviation \\
    \Centerstack{Classical Bayesian} & \Centerstack{1.317} & \Centerstack{0.006} \\
    \Centerstack{Score function}   & \Centerstack{1.318} & \Centerstack{0.005} \\
    [2ex] \toprule
    \end{tabular}
    }
    \caption{Means and deviations for posterior samples of $\hattheta_{1100}$.}
    \label{tab:Gaussian_covariance}
\end{table}

\vspace{0.5 em}
\noindent
The results are shown in Figure \ref{pl:Gaussian_covariance} and Table \ref{tab:Gaussian_covariance}. We can see the results are also quite close to each other.

\subsubsection{Autoregressive flow model}
Now we do the experiments on the autoregressive flow models {which gives a general scheme for defining expressive probability distributions by simple and basic distributions}; see \cite{Rezende_2015}, \cite{Kingma_2016} and~\cite{Papamakarios_2021}. In the autoregressive flow models, using gradients to do optimization is inevitable. So it will be very natural and convenient to extend the procedure to a further uncertainty quantification by our {sampling algorithm for the martingale posterior}.

Assume $z \sim p_z(z)$, $z \in \mathbb{R}^d$. $p_z$ is often a very simple distribution, called the base distribution. Now set $d=5$. We consider the inverse version of the autoregressive flow model with the affine transformers and the linear conditioners. It is very easy to see, no matter forward or inverse version of the model, if the base distribution is Gaussian, these are both Cholesky decomposition for the covariance matrices. In detail, consider
\begin{align*}
x_0&={a_0}z_0\\
x_1&={a_1}z_1+b_{10}x_0\\
x_2&={a_2}z_2+b_{21}x_1+b_{20}x_0\\
x_3&={a_3}z_3+b_{32}x_2+b_{31}x_1+b_{30}x_0\\
x_4&={a_4}z_4+b_{43}x_3+b_{42}x_2+b_{41}x_1+b_{40}x_0\\
\end{align*}
We denote $x=\mathcal{T}(z,\theta)$. We omit $\theta$ when no misleading. So $p_x(x)=p_z(\mathcal{T}^{-1}(x))|\det(J_\mathcal{T})|^{-1}$. The inverse of a triangular operator is also triangular. And here $\det(J_\mathcal{T})=\prod a_i$ is the Jacobian determinant of $\mathcal{T}$, which is very easy to compute. So $\log p_{x,\theta}(x)=\log p_z(\mathcal{T}^{-1}(x))-\log |\det(J_\mathcal{T})|.$ However, under nonlinear transformer cases, we need to notice the calculation. If we know $J_{\mathcal{T}^{-1}}(x)$, we may directly find the gradient $\nabla_{\theta}\log |\det J_{\mathcal{T}^{-1}}(x)|$. Otherwise, if we just know $J_\mathcal{T}(z)$, we should notice
\begin{align}
\label{equ: gradient_Jacobian}
\nabla_{\theta}\log \det J_{\mathcal{T}^{-1}}(x)=-[\nabla_{\theta}\log|\det J_\mathcal{T}(z)|+\nabla_z \log |\det J_\mathcal{T}(z)|\cdot\frac{\text{D} \mathcal{T}^{-1}(x)}{\text{D} \theta}].
\end{align}
To simplify the notation, all vectors in (\ref{equ: gradient_Jacobian}) are row vectors. Here ${\text{D} \mathcal{T}^{-1}(x)}/{\text{D} \theta}$ is the Jacobian matrix of $\mathcal{T}^{-1}$ with respect to the parameter $\theta$, which can be computed by ${\text{D} \mathcal{T}^{-1}(x)}/{\text{D} \theta}=-({\partial \mathcal{T}}/{\partial z})^{-1}{\partial \mathcal{T}}/{\partial \theta}=-J_\mathcal{T}^{-1}{\partial \mathcal{T}}/{\partial \theta}.$
In normalizing flows, we often choose $\mathcal{T}$ such that $\det(J_\mathcal{T})$ is very easy to compute. For $\log p_z(\mathcal{T}^{-1}(x))$, we have $\nabla_{\theta} \log p_z(\mathcal{T}^{-1}(x))=\nabla_{z} \log p_z(z) \cdot {\text{D} \mathcal{T}^{-1}(x)}/{\text{D} \theta}.$  Therefore, to implement our algorithm, we just need to know $z=\mathcal{T}^{-1}(x)$, which is what we directly generate in each iteration, $J_\mathcal{T}$ and ${\partial \mathcal{T}}/{\partial \theta}$. We do not need to use any more information about $\mathcal{T}^{-1}$.

However, with linear conditioners, this is a special case where we can write the form of the gradients explicitly as in (\ref{equ: gradient_Jacobian}). In many cases, especially when the model is large and deep, for example the diffusion models which are quite popular these days, we can use the forward and backward propagation techniques by the Leibniz chain rule to calculate the gradients; see 
\cite{Linnainmaa_1976}, \cite{Lecun_1998} and \cite{Schmidhuber_2015}.

\vspace{0.5 em}
\noindent
In the following two experiments, we use different base distributions: standard Gaussian and standard Laplacian. And we set collected data size $n=10000$, {martingale posterior} sample size $D=1000$. For each martingale, the run size $T=1000$, that is we sample $D=1000$ of $\hattheta_{11000}$. The true values for $a_i$ are all 2 and $b_{ij}$ are all 1.

\subsubsection{Autoregressive flow with Gaussian base distribution}
We do initialization using linear regression techniques due to the affine transformer and the linear conditioner setting here. The results are shown in Figure \ref{pl:Autoregressive_gaussian} and Table \ref{tab:Autoregressive_gaussian}.

\begin{figure}[ht!]
\centering
\subfigure[]{
\label{pl:Autoregressive_gaussian_a1}
\includegraphics[width=0.45\textwidth]{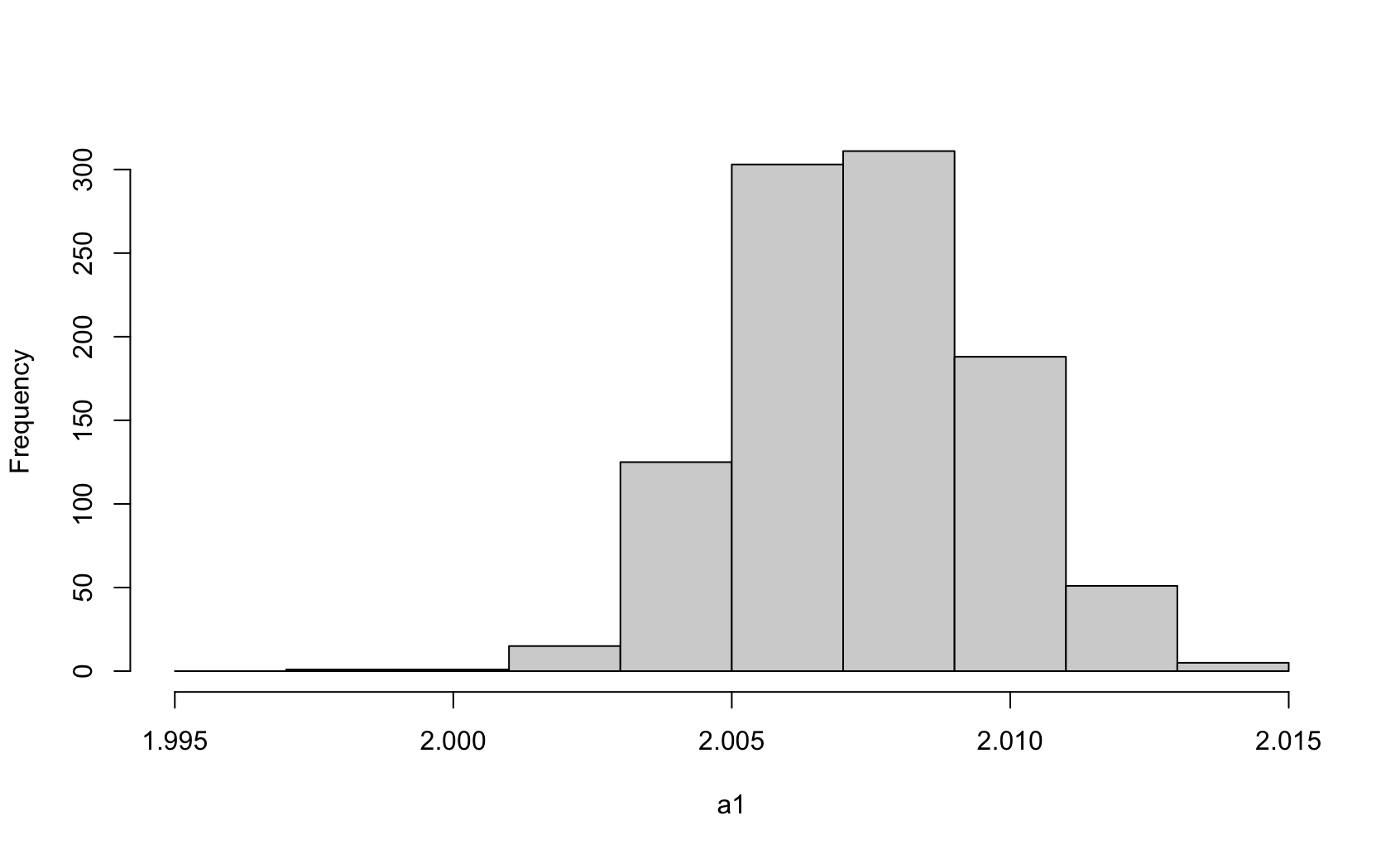}}
\subfigure[]{
\label{pl:Autoregressive_gaussian_b31}
\includegraphics[width=0.45\textwidth]{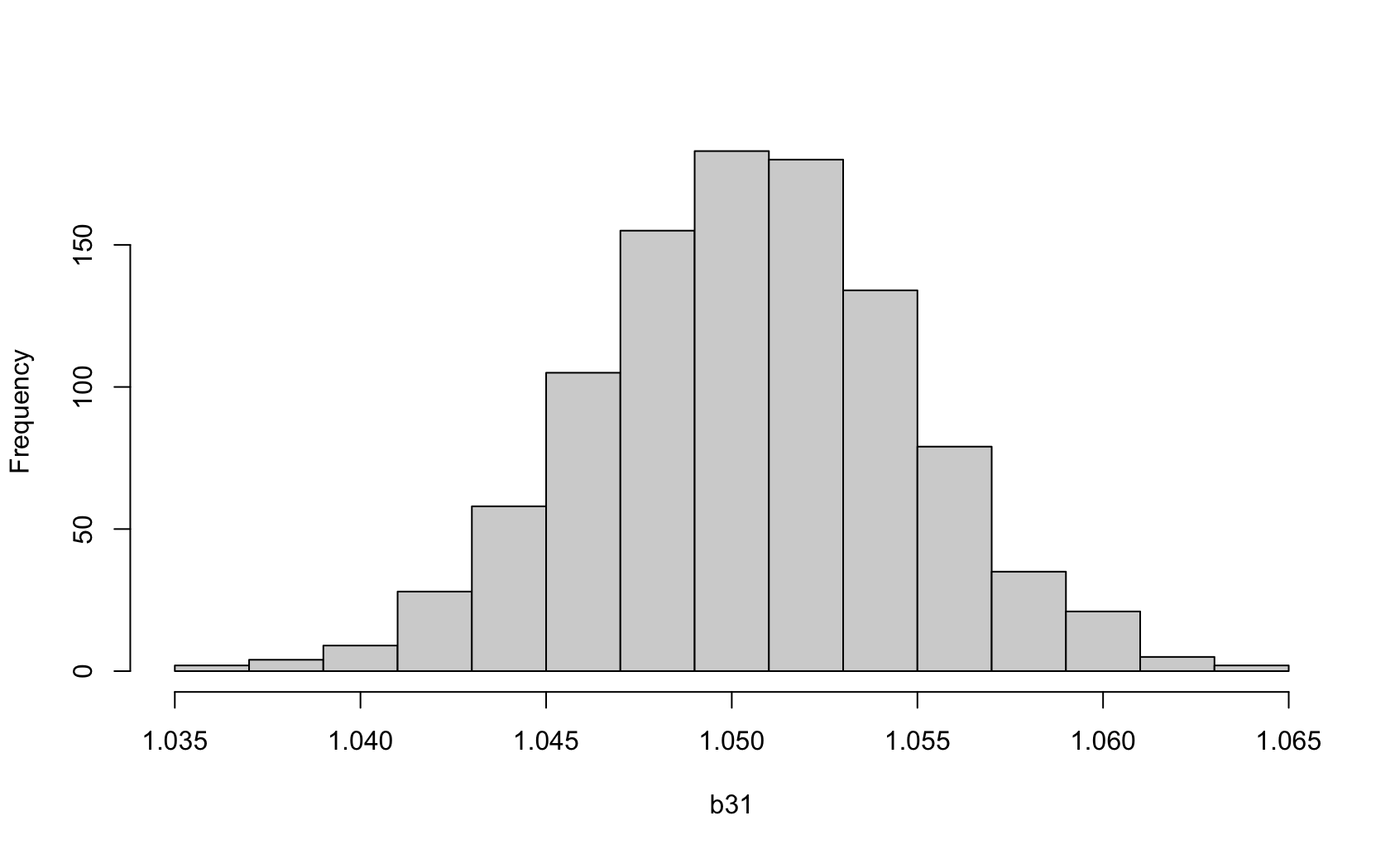}}
\caption{Histograms of posterior samples for autoregressive flows with standard Gaussian base distribution: (a) $a_1$; (b) $b_{31}$.}
\label{pl:Autoregressive_gaussian}
\end{figure}

\begin{table}[ht!]
    \centering
    {\renewcommand{\arraystretch}{2}
    \begin{tabular}{ccc}
    \bottomrule \hline
    Estimated element & Mean & Standard Deviation \\
    \Centerstack{$a_1$} & \Centerstack{2.0074} & \Centerstack{0.0022} \\
    \Centerstack{$b_{31}$}   & \Centerstack{1.0505} & \Centerstack{0.0043} \\
    [2ex] \toprule
    \end{tabular}
    }
    \caption{Means and deviations of posterior samples for autoregressive flows with standard Gaussian base distribution.}
    \label{tab:Autoregressive_gaussian}
\end{table}

\subsubsection{Autoregressive flow model with Laplace base distribution}
Linear models with Laplacian errors cannot express the solutions explicitly. Due to the properties of Laplacian, we imitate the initialization of Gaussian cases, except for just changing the initialization of $a_i$ as $\sqrt{{\text{MSE}}/{2}}$ where MSE is the mean squared error, because if \(X \sim \operatorname{Laplace}(\nu, 1)\) then \(k X+c \sim \operatorname{Laplace}(k \nu+c,|k| )\), and the variance is $2k^2$ for any constant $k$. The results are shown in Figure~\ref{pl:Autoregressive_laplace} and Table \ref{tab:Autoregressive_laplace}.

\begin{figure}[ht!]
\centering
\subfigure[]{
\label{pl:Autoregressive_laplace_a1}
\includegraphics[width=0.45\textwidth]{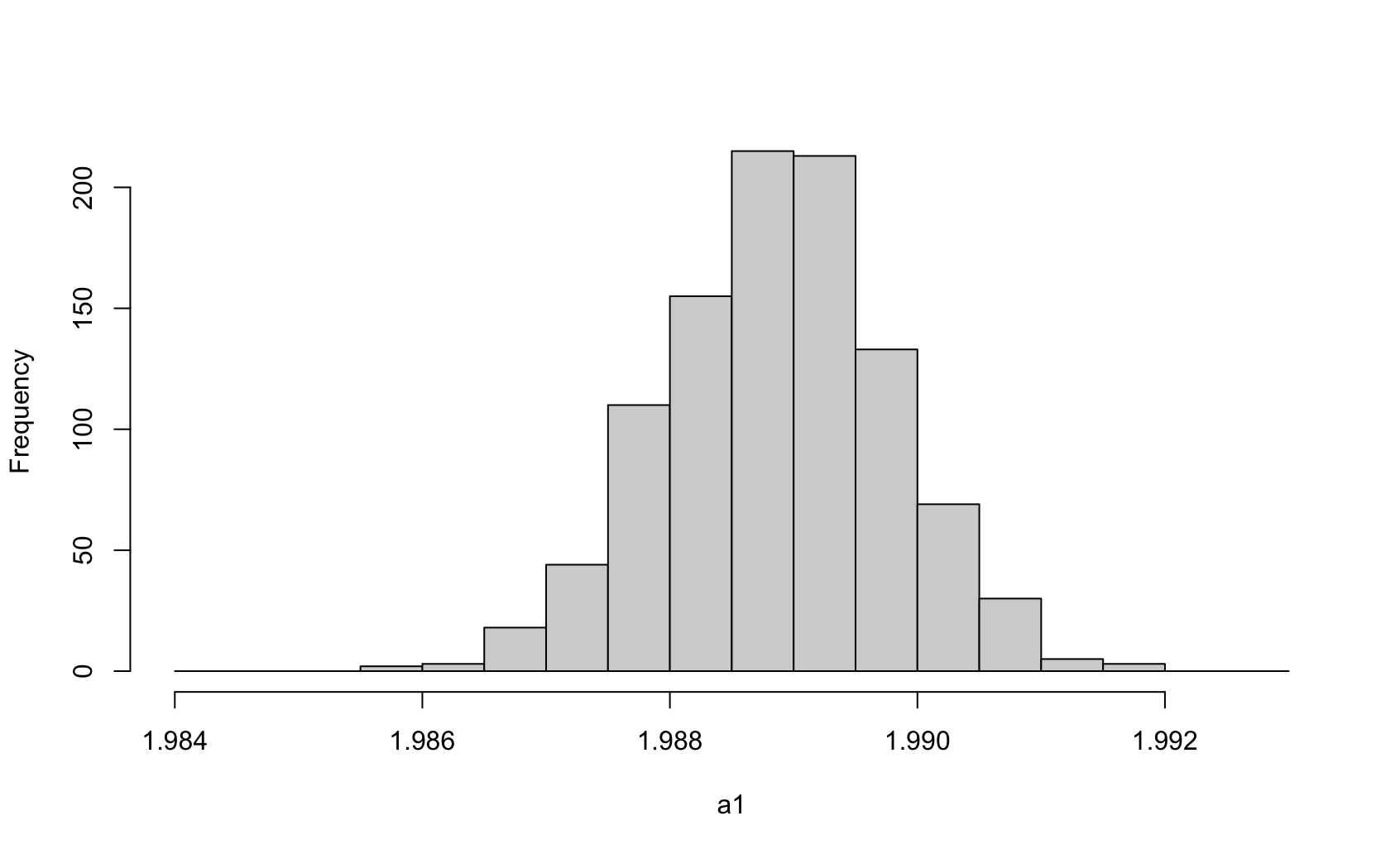}}
\subfigure[]{
\label{pl:Autoregressive_laplace_b31}
\includegraphics[width=0.45\textwidth]{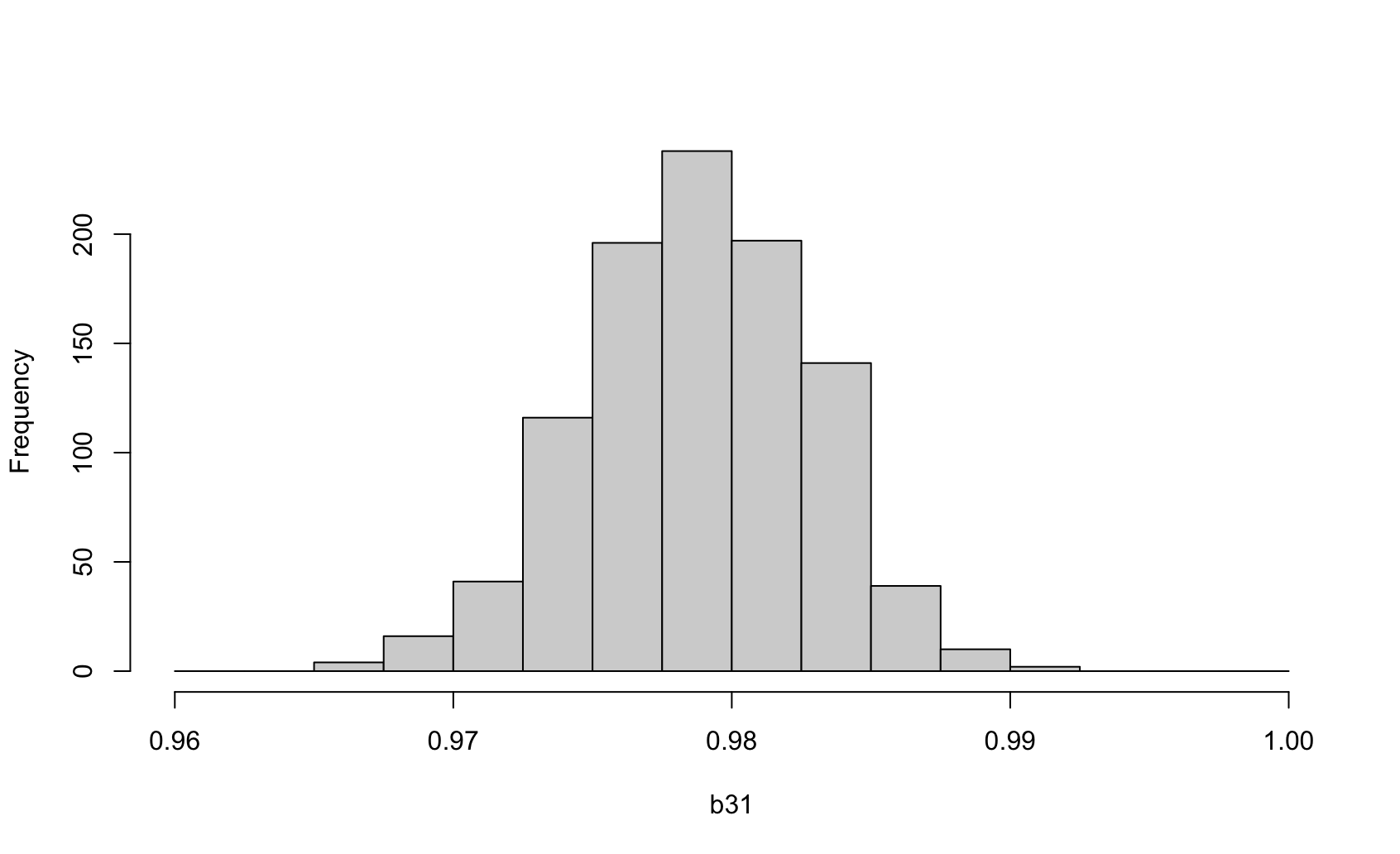}}
\caption{Histograms of posterior samples for autoregressive flows with standard Laplacian base distribution: (a) $a_1$; (b) $b_{31}$.}
\label{pl:Autoregressive_laplace}
\end{figure}

\begin{table}[ht!]
    \centering
    {\renewcommand{\arraystretch}{2}
    \begin{tabular}{ccc}
    \bottomrule \hline
    Estimated element & Mean & Standard Deviation \\
    \Centerstack{$a_1$} & \Centerstack{1.9889} & \Centerstack{0.0009} \\
    \Centerstack{$b_{31}$}   & \Centerstack{0.9788} & \Centerstack{0.0041} \\
    [2ex] \toprule
    \end{tabular}
    }
    \caption{Means and deviations of posterior samples for autoregressive flows with standard Laplacian base distribution.}
    \label{tab:Autoregressive_laplace}
\end{table}

In this case, {if we update the parameters of interest using a posterior mean according to the Doob's martingale convergence theorem}, we need to do linear regression to calculate MSE multiple times. However, with our methods, we only need to calculate them once.

\subsection{Comparison with variational Bayes -- simulated data}
\label{subsec:Comparison with variational Bayes -- simulated data}
We implemented comparisons between our score function method and VB, also known as variational inference. VB is an natural extension of the EM algorithm {with a wide range of applications, for instance, in reinforcement learning,} and is well-studied in the literature; for example, see~\cite{Ishiguro_2017}, \cite{Zintgraf_2021} and \cite{vanNiekerk_2024}. In summary, instead of finding the exact maximum likelihood estimates (MLEs), VB optimizes the Kullback-Leibler (KL) divergence between the proposed family of approximation distributions and the posterior distribution. By using simple approximation distributions and making certain assumptions, such as the independence of parameters, VB simplifies the computations. One of the main reasons for choosing VB is that the classical prior-posterior Bayesian often requires Markov Chain Monte Carlo (MCMC), which leads to high computational costs. However, VB is not unbiased and lacks accuracy. Therefore, using VB represents a trade-off: better speed at the cost of accuracy.

On the other hand, our score function method can sample from posterior distributions and perform UQ also without relying on MCMC. Moreover, the score function method is unbiased and can be implemented in parallel. Thus, given sufficient computational resources, the score function method can be both fast and ensure accuracy and unbiasedness.

Even if we use a small model, we can illustrate the above claim. We simulate $n=100$ data from $\mathcal{N}(0,5^2)$, and compare our score function method with VB for finding the posterior distributions of the mean $\mu$ and the variance $\sigma^2$. The results are shown in Figure~\ref{pl:Comnparison with VB_simulated} and Table~\ref{tab:comp_VB}. The left one shows the posteriors for $\mu$, and the right one shows those for $\sigma^2$. The black curves are the VB results by using Algorithm 1 in~\cite{Tran_2021}. The histograms are the score function method results. The gray vertical dotted lines are the averages for the martingale posterior samples. We can see that their posteriors are quite similar. 

\begin{figure}[ht!]
\centering
\subfigure[]{
\includegraphics[width=0.45\textwidth]{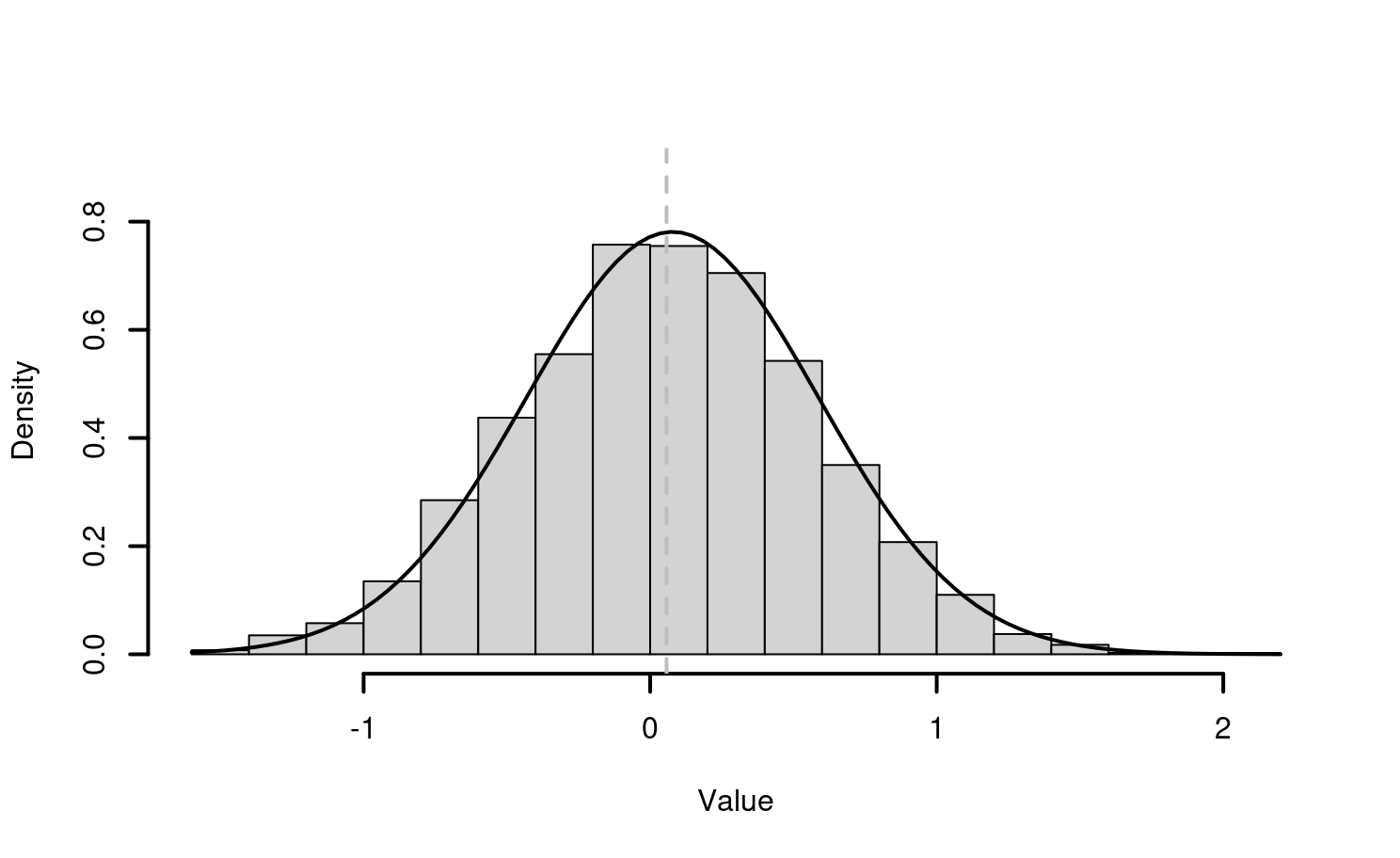}}
\subfigure[]{
\includegraphics[width=0.45\textwidth]{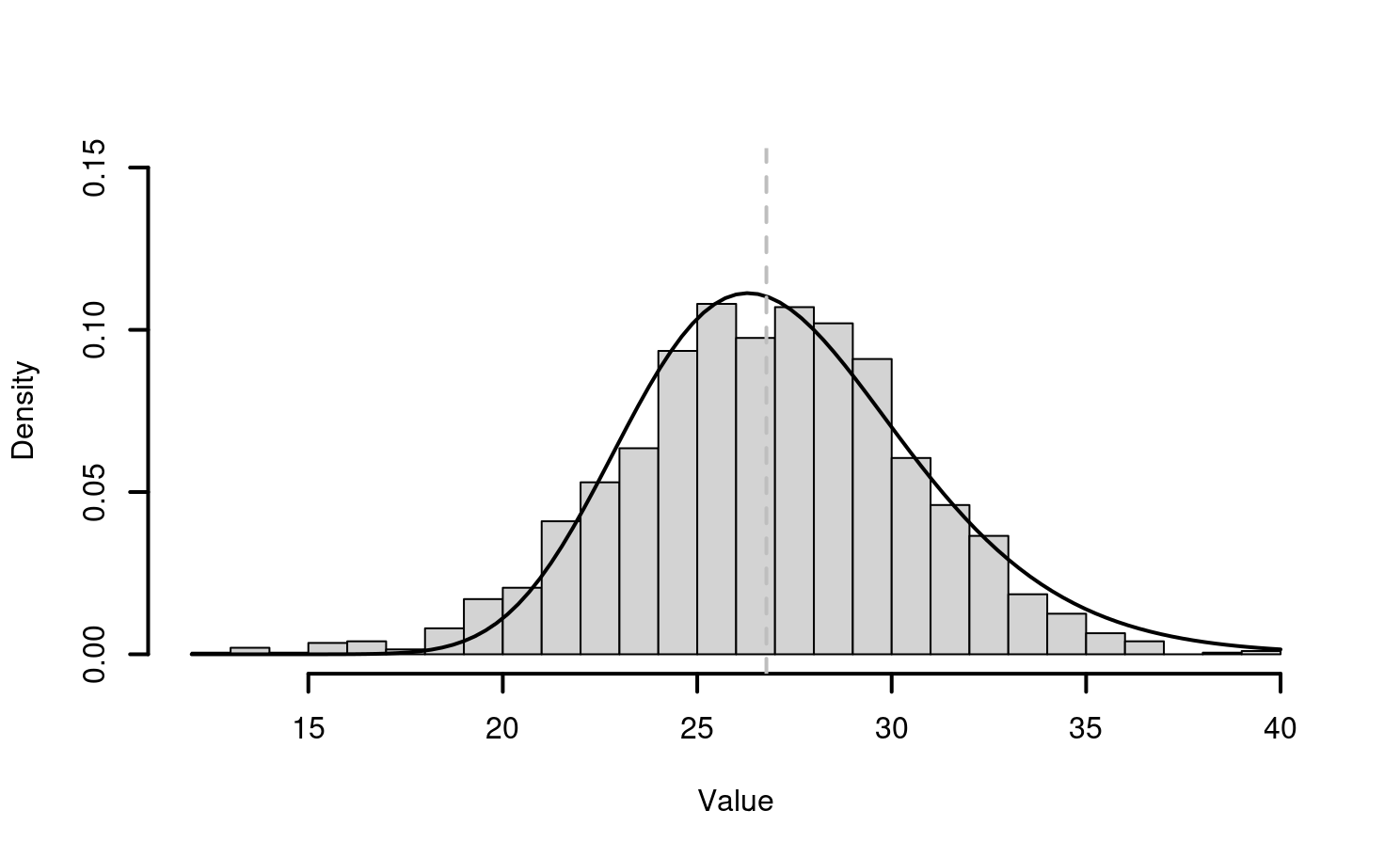}}
\caption{Comparison with variational Bayes (black lines) and martingale posterior (histogram) using the simulated data.}
\label{pl:Comnparison with VB_simulated}
\end{figure}

\begin{table}[ht!]
    \centering
    {\renewcommand{\arraystretch}{2}
    \begin{tabular}{ccccc}
    \bottomrule \hline
    Method & $\mu$ post-mean& $\mu$ post-variance&$\sigma^2$ post-mean& $\sigma^2$ post-variance\\
    \Centerstack{VB} & \Centerstack{0.078} & \Centerstack{0.261} &27.276&14.037 \\
    \Centerstack{Score function}   & \Centerstack{0.057} & \Centerstack{0.267}&26.778&13.980 \\
    [2ex] \toprule
    \end{tabular}
    }
    \caption{Means and deviations for posterior samples of $\hattheta_{5100}$.}
    \label{tab:comp_VB}
\end{table}

\subsection{Comparison with variational Bayes -- real data}
\label{subsec:Comparison with variational Bayes -- real data}

In this section, we compare our score function method with variational Bayes on a real dataset. \cite{Bates_1988} applies a nonlinear model
$$
x_i = b_1\{1-\exp(-b_2t_i)\}+e_i
$$
on the data from~\cite{Marske_1967} about biochemical oxygen demand $x$ of prepared water samples to incubation time $t$, where $e_i$ is assumed to be the i.i.d. Gaussian distributed error with the variance $\sigma^2$ to be estimated. The dataset is quite small with only 6 data points: 
$$\left\{\left(t_{i}, x_{i}\right)\right\}=\{(1,8.3),(2,10.3),(3,19),(4,16),(5,15.6),(7,19.8)\}.$$
There are a lot of literature dealing with this problem; for example, see~\cite{Newton_1994} about a weighted likelihood bootstrap method.

In Section~\ref{subsec:Comparison with variational Bayes -- simulated data} for a simulated dataset, we use mean field variational Bayes (MFVB). However, MFVB assumes the independence between every parameter. {Even if MFVB provides good posterior mean estimates, it often underperforms on variances and covariances, and} the deeper relationship between the parameters cannot be caught {; see~\cite{Giordano_2018} for example}. Therefore, here we use fixed form variational Bayes (FFVB). Instead of assuming the independence, FFVB assumes a fixed parametric family for the approximation density. For example, if we use multivariate Gaussian as the approximation density, the mean vector and the covariance matrix are the parameters to be estimated. More details can be found in~\cite{Tran_2021}. {See~\cite{Honkela_2010} for an application of FFVB by the gradient method on Riemannian manifolds.}

In our experiment, we use R package ‘LaplacesDemon’~\cite{LaplacesDemon_R_package} to implement FFVB, with the normal priors for $b_1$ and $b_2$ and the inverse gamma prior for $\sigma^2$. Figure~\ref{pl:Comnparison with VB_real} shows $D=1000$ posterior samples with their empirical density contours for each method. The left panel is from the score function method which are very similar to the results in~\cite{Newton_1994}, while the right is from the FFVB method. In the left panel, we update the parameters by
$$
\theta_{m+1}=\theta_m + H_m^{-1} s(\theta_m, x_{m+1}),
$$
where $\theta_m = (b_{1,m},b_{2,m})^{\top}$ and $H_m$ is the matrix of second derivatives at $\theta_m$ and summed over the $x_{1:m}$. The score function and Hessian are the vector of first derivatives and matrix of second derivatives arising from the log-density function, 
$$\left[-\frac{1}{2}\log \sigma-\frac{1}{2}(x-b_1(1-e^{-b_2t}))^2/\sigma^2\right].$$

\begin{figure}[ht!]
\centering
\subfigure[]{
\includegraphics[width=0.45\textwidth]{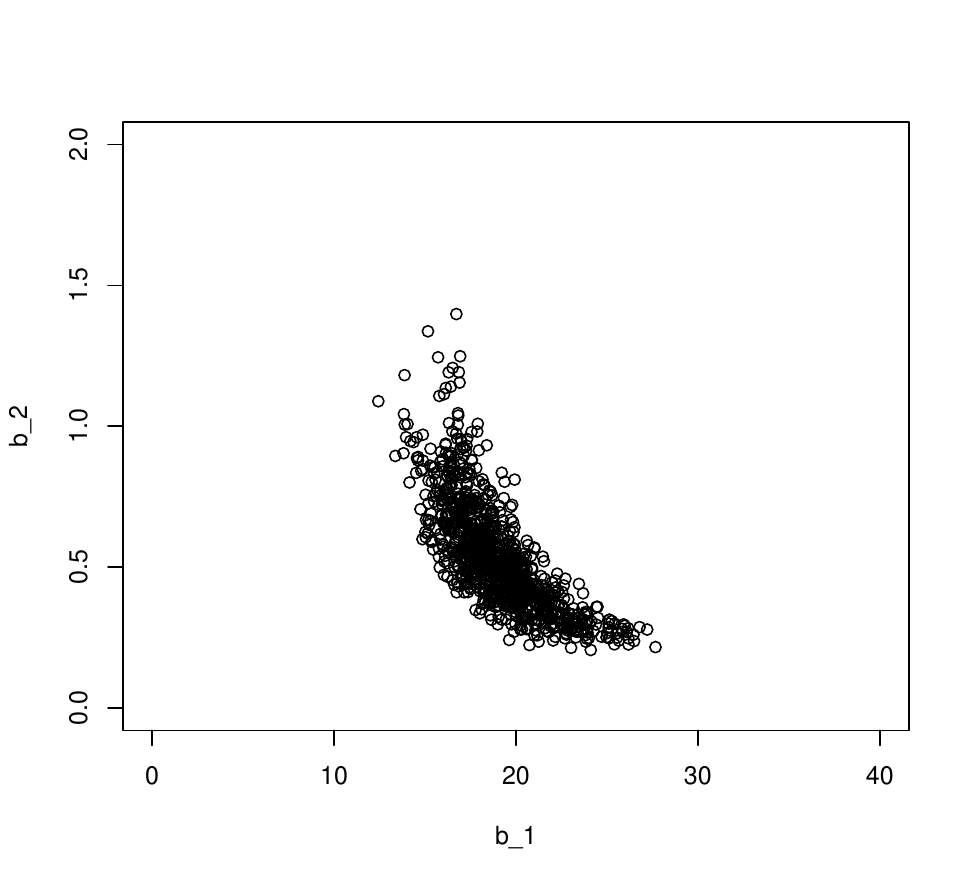}}
\subfigure[]{
\includegraphics[width=0.49\textwidth]{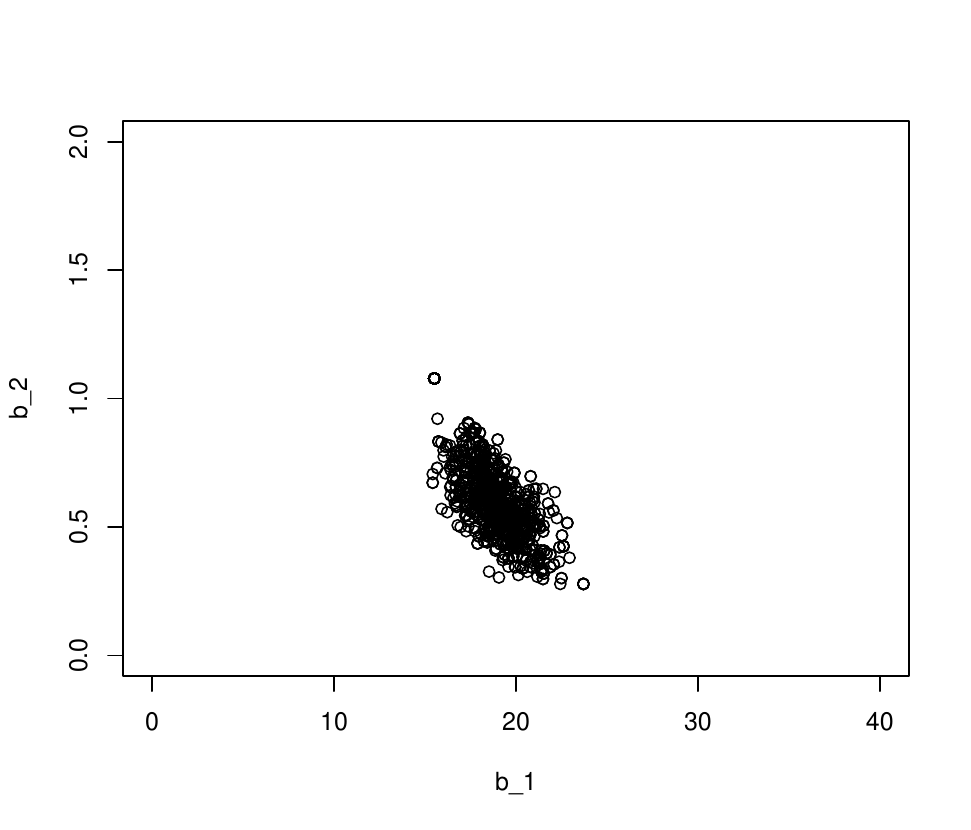}}
\caption{Comparison with variational Bayes with the real data: two left panels corresponding to different settings of the martingale posterior while the right panel uses a mean field variational Bayesian approach.}
\label{pl:Comnparison with VB_real}
\end{figure}

Compared with the natural martingale posterior computed by the score function method, VB needs an artificially assigned approximation density family, leading to an unnatural, constrained and approximate relationship between each parameter. 
Further, VB provides biased samples and takes substantially more time to run compared to the score function method.

\section{Summary and discussion}
\label{sec:summary_and_discussion}

Bayesian uncertainty arises from the unseen data. To quantify this uncertainty the need is to model the unseen data given what has been seen. This is done using one step ahead predictive density functions, as originally described in \cite{Fong_2021}. Martingales play a key role and therefore how to construct appropriate martingales becomes the key problem. In this paper, we discuss the martingale posteriors constructed from score functions
by considering a sequence of estimators based on increasing population sizes. A Bayesian posterior arises in the limit as the population size tends to infinity. The convergence of the parameter ensures the sequence of missing observations are asymptotically exchangeable which in turn guarantees the existence of a posterior distribution. This is in sharp contrast with the Bayesian paradigm which insists on the observed and missing observations being exchangeable, which introduces constraints which makes inference often challenging. On the other hand, within our asymptotically exchangeable setting, we can run martingales in parallel and only need to sample the model which is typically straightforward to do.

Due to the zero expectation property of score functions, the construction is highly suited to martingales. The technical aspects of our paper are then concerned with convergence and asymptotic exchangeability in one- and multi- dimensional cases. 
The structure of the martingale algorithm is similar to gradient-based optimization algorithms. This matches the idea of getting the best estimator for the next step. 
The algorithm can be easily implemented even for complicated models, as we have demonstrated in a number of illustrations. 

\bibliographystyle{abbrv}
\bibliography{Fuheng}
\end{document}

%% file: final_macros.tex
\DeclareMathOperator{\diag}{diag}

\DeclareMathOperator{\trace}{trace}

\newtheoremstyle{named}{}{}{\itshape}{}{\bfseries}{.}{.5em}{\thmnote{#3's }#1}
\theoremstyle{named}

\theoremstyle{plain}

\newtheorem{theorem}{Theorem}

\newlength{\widebarargwidth}
\newlength{\widebarargheight}
\newlength{\widebarargdepth}

\makeatletter
\long\def\@makecaption#1#2{
        \vskip 0.8ex
        \setbox\@tempboxa\hbox{\small {\bf #1:} #2}
        \parindent 1.5em  
        \dimen0=\hsize
        \advance\dimen0 by -3em
        \ifdim \wd\@tempboxa >\dimen0
                \hbox to \hsize{
                        \parindent 0em
                        \hfil
                        \parbox{\dimen0}{\def\baselinestretch{0.96}\small
                                {\bf #1.} #2
                                }
                        \hfil}
        \else \hbox to \hsize{\hfil \box\@tempboxa \hfil}
        \fi
        }
\makeatother

\long\def\comment#1{}
\definecolor{battleshipgrey}{rgb}{0.52, 0.52, 0.51}
\definecolor{darkgray}{rgb}{0.66, 0.66, 0.66}
\definecolor{darkgreen}{rgb}{0.0, 0.2, 0.13}
\definecolor{darkspringgreen}{rgb}{0.09, 0.45, 0.27}
\definecolor{dukeblue}{rgb}{0.0, 0.0, 0.61}
\definecolor{olivedrab7}{rgb}{0.24, 0.2, 0.12}
\definecolor{darkblue}{rgb}{0.0, 0.0, 0.55}
\definecolor{darkscarlet}{rgb}{0.34, 0.01, 0.1}
\definecolor{candyapplered}{rgb}{1.0, 0.03, 0.0}
\definecolor{ao(english)}{rgb}{0.0, 0.5, 0.0}
\definecolor{applegreen}{rgb}{0.55, 0.71, 0.0}
